\documentclass{article}

\usepackage[dvipsnames]{xcolor}
\usepackage{microtype}
\usepackage{enumerate}
\usepackage{graphicx}
\usepackage{balance}
\usepackage{caption}
\usepackage{bm}
\usepackage{subcaption}
\usepackage{booktabs} % for professional tables
\usepackage{amsmath}
\usepackage{amssymb}
\usepackage{mathtools}
\usepackage{amsthm}
\usepackage{paralist}
\usepackage{enumitem}
\usepackage{tabularx}

\usepackage{hyperref}

 \usepackage[accepted]{icml2024}

\newif\ifnotes
\notestrue
\ifnotes
\newcommand{\sam}[1]{{\ifnotes \textcolor{Thistle}{\scriptsize Sam: {#1}} \fi}}
\newcommand{\shikha}[1]{\ifnotes {\noindent \scriptsize  \textcolor{blue} {Shikha: {#1}}} \fi{}}
\newcommand{\ben}[1]{{\ifnotes \scriptsize \textcolor{orange}{Ben: {#1}} \fi}}
\newcommand{\aidin}[1]{{\ifnotes \scriptsize \textcolor{red}{Aidin: {#1}} \fi}}

\else
\newcommand{\sam}[1]{}
\newcommand{\shikha}[1]{}
\newcommand{\ben}[1]{}
\newcommand{\aidin}[1]{}

\fi

\theoremstyle{plain}
\newtheorem{theorem}{Theorem}[section]

\newtheorem{lemma}[theorem]{Lemma}
\newtheorem{invariant}[theorem]{Invariant}

\theoremstyle{definition}

\theoremstyle{remark}

\renewcommand{\tilde}{\widetilde}
\newcommand{\heta}{\hat{\eta}}
\newcommand{\calH}{\mathcal{H}}
\newcommand{\anp}{\tilde{\alpha}}
\newcommand{\dep}{\tilde{\delta}}

\newcommand{\defn}{\emph}
\newcommand{\etamax}{\eta}
\renewcommand{\epsilon}{\varepsilon}
\renewcommand{\defn}[1]{\emph{\textbf{#1}}}
\renewcommand{\paragraph}[1]{
\noindent{\textbf{#1}}~}

\begin{document}
\sloppy 
\twocolumn[
\icmltitle{Incremental Topological Ordering and Cycle Detection with Predictions}



%\icmlsetsymbol{dagger}{$\dagger$}
\begin{icmlauthorlist}
\icmlauthor{Samuel McCauley}{williams}
\icmlauthor{Benjamin Moseley}{cmu}
\icmlauthor{Aidin Niaparast}{cmu}
\icmlauthor{Shikha Singh}{williams}
\end{icmlauthorlist}

\icmlaffiliation{williams}{Department of Computer Science, Williams College, Williamstown, MA 01267 USA}
\icmlaffiliation{cmu}{Tepper School of Business, Carnegie Mellon University,Pittsburgh, PA 15213 USA}

\icmlcorrespondingauthor{Samuel McCauley}{sam@cs.williams.edu}
\icmlcorrespondingauthor{Benjamin Moseley}{moseleyb@andrew.cmu.edu}
\icmlcorrespondingauthor{Aidin Niaparast}{aniapara@andrew.cmu.edu}
\icmlcorrespondingauthor{Shikha Singh}{shikha@cs.williams.edu}

\icmlkeywords{Machine Learning, ICML}

\vskip 0.3in
]

\newif\ifproofof
% \proofoffalse
\proofoftrue

\ifproofof
\newenvironment{proofof}[1]{$ $\newline \noindent{\em Proof of {#1}. }\ignorespaces}{\qed \newline}

\else
\newenvironment{proofof}[1]{\noindent{\em Proof. }\ignorespaces}{\qed \newline }
\fi



\begin{NoHyper} %\sam{avoids hyperref warning}
% UNCOMMENT THE TXT IN ICML APPEARING IN STYLE FILE FOR CAMERA READY
%\newcommand{\ICML@appearing}{

\printAffiliationsAndNotice{}  % leave blank if no need to mention equal contribution
\end{NoHyper}

\begin{abstract}

This paper leverages the framework of algorithms-with-predictions to design data structures for two fundamental dynamic graph problems: incremental topological ordering and cycle detection.  In these problems, the input is a directed graph on $n$ nodes, and the $m$ edges arrive one by one.  The data structure must maintain a topological ordering of the vertices at all times and detect if the newly inserted edge creates a cycle.  The theoretically best worst-case algorithms for these problems have high update cost (polynomial in $n$ and $m$).  In practice, greedy heuristics (that recompute the solution from scratch each time) perform well but can have high update cost in the worst case. 

In this paper, we bridge this gap by leveraging predictions to design a learned new data structure for the problems. 
Our data structure guarantees consistency, robustness, and smoothness with respect to predictions---that is, it has the best possible running time under perfect predictions, never performs worse than the best-known worst-case methods, and its running time degrades smoothly with the prediction error.  Moreover, we demonstrate empirically that predictions, learned from a very small training dataset, are sufficient to provide significant speed-ups on real datasets. 
\end{abstract}

%\vspace*{-.25in}
\section{Introduction}
 
A recent line of research 
has focused on how learned predictions can be used to enhance the running time of algorithms. This novel approach, often referred to as \emph{warm starting}, initializes an algorithm with a machine-learned starting state to optimize efficiency on a new problem instance. 
This starting state can significantly improve performance over the conventional method of solving problems from scratch.

Warm starting algorithms with machined-learned predictions can be viewed through the lens of {beyond-worst-case} analysis. While the predominant algorithmic paradigm for decades has been to use worst-case analysis, warm starting takes into account that real-world applications repeatedly solve a problem on similar instances that share a common underlying structure.  Predictions about these input instances can be used to the improve running time of future computations.

This new line of research, called \emph{algorithms with predictions} or \emph{learning-augmented algorithms}, leverages predictions to achieve strong guarantees--much like those achieved using worst-case analysis---for warm-started algorithms. Under this setting, the performance of the algorithm is measured as a function of the \emph{prediction quality}. This ensures that the algorithm is robust to prediction inaccuracies and has performance that interpolates smoothly between ideal and worst-case guarantees with respect to predictions. 

Recent proof-of-concept results have demonstrated the potential to enhance the running time of offline algorithms. The area was empirically initiated by Kraska et al. \cite{KraskaBCDP18}.  Theoretically, Dinitz et al. \cite{DinitzILMV21} were the first to provide a theoretical framework for using warm-start to improve the running time of the weighted bipartite matching problem. Follow-up works include the application of learned predictions to improve the efficiency of computing flows using Ford-Fulkerson~\cite{DaviesMVW}, shortest path computations using Bellman-Ford~\cite{LattanziSV23}, binary search~\cite{BaiC23}, convex optimization~\cite{SakaueO22} and maintaining a dynamic sorted array \cite{McCauleyMNS23}.
 These results showcase the promising potential to harness predictions more broadly for algorithmic efficiency.  

 Data structures are one of the most fundamental algorithmic domains, forming the backbone of most computer systems and databases.  Leveraging predictions to improve data structure design remains a nascent research area. Empirical investigations, initiated by~\cite{KraskaBCDP18} and follow-ups such as~\cite{ferragina2021performance}, demonstrate the exciting potential of speeding up indexing data structures using machine learning. More recently~\cite{McCauleyMNS23} developed the first data structure in the new theoretical framework of algorithms with predictions.  They design a learned data structure to maintain a sorted array efficiently under insertions (aka \emph{online list labeling}).
 Since then, two concurrent works~\cite{BrandFNPS24} and~\cite{HenzingerLSSY24} show how to leverage predictions for maintaining dynamic graphs for problems such as shortest paths, reachability, and triangle detection via predictions for the matrix-vector multiplication problem.

This paper focuses specifically on developing the area of data structures for dynamic graph problems. We study the fundamental problems of maintaining an \textbf{incremental topological ordering} of the nodes of a directed-acyclic graph (DAG) and the related problem of \textbf{incremental cycle detection}.  In the problem, a set of $n$ nodes $V$ is given and the edge set is initially empty.  Over time, {directed edges}  arrive that are added to the graph.  The algorithm must maintain a topological ordering of $V$ at all times. A \defn{topological ordering} is a labeling $L: V \rightarrow \mathbb{Z}$ of the vertices $V$ such that $L(v) < L(u)$ if there is a directed path from $v$ to $u$. A topological ordering exists if and only if the directed graph is acyclic.  Thus, if an edge is inserted that creates a cycle, the data structure must report that a cycle has been detected, after which the algorithm ends.

 The goal is to design an online algorithm that has small \defn{total update time} for the $m$ edge insertions.  Offline, when all edges are available a priori,  the problem can be solved in   $O(m)$ (linear time) by running Depth-First-Search (DFS). A naive approach to the incremental problem is to use DFS from scratch each time an edge arrives, giving $O(m^2)$ total time.  The goal is to design dynamic data structures that can perform better than this naive approach.

Topological ordering and cycle detection are foundational textbook problems on DAGs. Incremental maintenance of DAGs is ubiquitous in database and scheduling applications with dependencies between events (such as task scheduling, network routing, and casual networks).   
Due to their wide use, there has been substantial prior work on maintaining incremental topological ordering in the worst case (without predictions). Prior work can roughly be partitioned into the cases where the underlying graph is sparse or dense.   A line of work~\cite{BenderFiGi09,haeupler2012incremental,BenderFiGi15,BernsteinChechi18,BhattacharyaKulkarni20} for sparse graphs  led to~\cite{BhattacharyaKulkarni20} giving a randomized  algorithm with total update time $\tilde{O}(m^{4/3})$.  The  $\tilde{O}$ suppresses logarithmic factors. For dense graphs, a line of work \cite{CohenFKR13,BenderFiGi15} has total update time  $\tilde{O}(n^2)$. These results hold for both incremental topological ordering and cycle detection.  A recent breakthrough~\cite{ChenKyLi23} uses new techniques to improve the running time of incremental cycle detection to $O(m^{1+o(1)})$; their results do not extend to topological ordering. At present there are no nontrivial lower bounds for either problem, that is, it is not known if there exists an algorithm with update time $\tilde{O}(m)$.  

Despite the rich theoretical literature on the problem, there is limited empirical evidence of their success~\cite{ajwani2008n}.  As most practical data is non-worst-case, greedy brute-force methods do well empirically~\cite{baswana2018incremental}. The algorithms-with-predictions framework is motivated precisely by this disconnect between high-cost worst-case methods and simple practical heuristics.  The goal of designing learned algorithms in this framework is to extract beyond-worst-case performance on typical instances, while being robust to bad predictions in the worst case.

More formally, in the algorithms-with-predictions framework, an algorithm is (a) \defn{consistent} if it matches the offline optimal (or outperforms the worst case) under perfect predictions, (b) \defn{robust} if it is never worse than the best worst-case algorithm under adversarial predictions, and (c) \defn{smooth} if it interpolates smoothly between these extremes.  We call an algorithm \defn{ideal} if it is consistent, robust, and smooth.

 In this paper, we initiate the study of how learned predictions can be leveraged for incremental topological ordering. We propose a \defn{coarse-grained prediction model} and use it to design a new ideal data structure for the problem; see Section~\ref{sec:results}.  Moreover, we present a practical learned DFS algorithm and our experiments show that using even mildly accurate predictions leads to significant speedups. All our results extend the incremental cycle detection.  Our results complement the concurrent theoretical work by \cite{BrandFNPS24} on dynamic graph data structures; 
 see Section~\ref{sec:related}.

\subsection{Our Contributions}\label{sec:results}

We first propose a prediction model for the problem and then use it to formally describe our results.  

\paragraph{Coarse Prediction Model.}  
For the incremental topological ordering problem, it is natural to consider predictions on the nodes which give information about their relative ordering in the final graph.  Intuitively, a vertex is earlier in the ordering if it has few ancestors and many descendants.  For technical reasons, instead of predicting the number of ancestor and descendant vertices, we predict the number of ancestor and descendant edges.\footnote{This is because the running time of the learned algorithm depends on the number of edges traversed.}
More formally, for each vertex $v$, let $\alpha(v)$ be the total number of ancestor edges of $v$ after all edges arrive, and let $\delta(v)$ be the number of descendant edges.  An edge $(u,w)$ is an \textbf{ancestor edge} of $v$ in there is a directed path from $w$ to $v$.  The edge $(u,w)$ is a \textbf{descendent edge} of $v$ in there is a directed path from $v$ to $u$.  At the beginning of time, the algorithm is given predictions $\tilde{\alpha}(v)$ and $\tilde{\delta}(v)$ for $\alpha(v)$ and $\delta(v)$ for each vertex $v$.  The \emph{error} in the prediction for vertex $v$ is $\eta_v = |\alpha(v) - \tilde{\alpha}(v)| + |\delta(v) - \tilde{\delta}(v)|$.  The overall prediction error of the input sequence is $\etamax = \max_{v\in V} \eta_v$.\footnote{For simplicity, we assume throughout our analysis that $\etamax \geq 1$; this is to avoid $\etamax + 1$ terms throughout our running times.}  

We note that our prediction model predicts a small amount of information about the input, in contrast to models that predict the entire input sequence, e.g.~\cite{brand2023dynamic, henzinger2018decremental}.  In particular, predictions that predict the entire input are \defn{fine-grained}---each possible input sequence maps to a unique perfect prediction.  Our predictions are \defn{course-grained} because there are many possible input graphs that can map to a single perfect prediction.  Intuitively, the more coarse-grained the prediction, the more robust it is to small changes in the input.

\paragraph{Ideal Learned Ordering.} We present a new learned data structure for the incremental topological ordering, called \defn{Ideal Learned Ordering}. This data structure has total update time $\tilde{O}(\min\{n \etamax+m,  m\etamax^{1/3}, n^2 \})$. The data structure is ideal with respect to predictions, in particular, it is:
\begin{compactitem}
\item {\bf Consistent:} If $\etamax = O(1)$, its performance matches (up to logarithmic factors) the best possible running time $\tilde{O}(m)$ of an \emph{offline} optimal algorithm.  

\item {\bf Robust:} For any $\eta \leq m$,  the total running time is $\tilde{O}(\min\{m^{4/3}, n^2\})$, and thus its performance is never worse than the best-known worst-case algorithms~\cite{BenderFiGi15} and~\cite{BhattacharyaKulkarni20}.

\item {\bf Smooth:} For any intermediate error $\eta$, the performance smoothly interpolates as a function of $\eta$ (and $n$ and $m$), between the above two extremes.
\end{compactitem}

At a high level, the ideal learned ordering decomposes the vertices into subproblems based on the predictions.  On each subproblem, it runs the best-known worst-case algorithm, which is warm-started with the predictions.

\paragraph{Learned DFS Ordering and Empirical Results.} In addition to the above ideal algorithm, we present a simple practical data structure that essentially warm-starts depth-first search using predictions.  We call this the \defn{learned DFS ordering} (LDFS).  This data structure has total update time $O(m \etamax)$; thus each insert has running time $O(\etamax)$. We implement LDFS and our experiments show that with very little training data, the predictions deliver excellent speed-ups.  In particular,  we demonstrate on real time-series data that using only 5\% of training data, LDFS explores over 36x fewer vertices and edges than baselines, giving a 3.1x speedup in running time. Moreover, its performance is extremely robust to prediction errors; see Figure~\ref{fig:email-Eu-core robustness}. 

\subsection{Related Work}\label{sec:related}
Recently, \cite{BrandFNPS24} leverage predictions for dynamic graph data structures. They give a general result for the online matrix-vector multiplication problem where the matrix is given and a sequence of vectors arrive online. They apply this to several dynamic graph problems including cycle detection. 
Their data structure requires $O(n^\omega+ n \sum_{i \in V} \min \{ \delta_i,n \})$ total time where $\delta_i$ is the error between when edge $i$ arrives and when it is predicted to arrive and $n^\omega$ is the time to perform matrix multiplication.    Their prediction is the entire input, that is, the online sequence of vectors.  The predictions used in this work are more course-grained (only require a pair of numbers per vertex), and thus are robust to small perturbations to the input sequence.   
Moreover, their work is purely theoretical and leaves open (a) how predictions can be leveraged for maintaining topological ordering, and (b) how predictions can be empirically leveraged for dynamic graph problems. Our work addresses both and complements their findings.     

The Ideal Learned Ordering uses the best-known sparse algorithm~\cite{BhattacharyaKulkarni20} and the best-known dense algorithm~\cite{BenderFiGi15}, referred to as BK and BFGT throughout. We briefly summarize them; we refer to the papers for more details.   

The BFGT algorithm maintains levels $\ell(u)$ for each vertex $u$: these are underestimates of the total number of ancestors of $u$ in the final graph.  The levels are initially set to $1$. On an edge insertion $(x, y)$, if $\ell(x) > \ell(y)$, they greedily update levels to maintain a topological ordering. To improve the efficiency, they update $y$'s level even if $\ell(x) \leq \ell(y)$ if a better underestimate of the number of ancestors of $y$ is available (based on its predecessors' levels). The total time for all insertions is bounded by $\tilde{O}(m + \sum_v \ell(v))$. As $\ell(v)$ is at most the number of ancestors, their total running time is $\tilde{O}(n^2)$.  
In Section~\ref{sec:theory}, we use predictions to ensure that BFGT is run on subproblems containing vertices with $O(\eta)$ ancestors.
Thus, the levels can only increase at most $\eta$ times, which leads to the total update time $\tilde{O}(m + n \eta)$.

The BK algorithm (which is based on~\cite{BernsteinChechi18}) also partitions the vertices into levels, but these are based on their ancestors \emph{and} descendants. It is a randomized algorithm and initializes the vertex levels using sampling. In particular, they use an internal parameter $\tau$ where each vertex $v \in V$ is sampled with probability $\Theta(\log n/\tau)$.  A vertex is in a level $(i,j)$ if it has $i$ ancestors and $j$ descendants among the sampled nodes.  They bound the number of possible ancestors and descendants of a vertex within a level using the parameter $\tau$.  In Section~\ref{sec:theory}, we use their algorithm as a blackbox with the exception that we set $\tau$ using predictions.  For the analysis, we give a tighter bound of Phase I and II of their algorithm.

\subsection{Organization}\label{sec:outline}
Section~\ref{sec:model} defines the model. Learned DFS Ordering is presented in Section~\ref{sec:dfs}; which is generalized to the Ideal Learned Ordering in Section~\ref{sec:theory}.
Finally, Section~\ref{sec:experiments} presents experimental results. For space, many proofs and further experiments are deferred to the Appendices~\ref{sec:omitted_proofs} and \ref{sec:additional_experiments}.

\section{Model and Definitions}
\label{sec:model}

\paragraph{Directed Graphs.}
Consider a directed graph $G = (V,E)$ with $|V| = n$ and $|E| =m$.  Let $(u,v) \in E$ denote a directed edge from $u$ to $v$.  We say that a vertex $v$ is an \defn{ancestor} of vertex $w$ if there is a path from $v$ to $w$ in the graph. We say $w$ is a \defn{descendant} of $v$ if $v$ is an ancestor of $w$.  A vertex is an ancestor and descendant of itself.  We say that an edge $(u, v)$ is an \defn{ancestor edge} of a vertex $w$ if $v$ is an ancestor of $w$.  Similarly, an edge $(u,v)$ is a \defn{descendant edge} of a vertex $w$ if $u$ is a descendant of $w$. If $(v, w)$ is an edge then we say that $v$ is a \defn{parent} of $w$ and $w$ is a \defn{child} of $v$.  
A \defn{topological ordering} of a directed acyclic graph (DAG) $G = (V, E)$ is a labeling $L: V \rightarrow \mathbb{Z}$ such that for every edge $(v, w) \in E$, we have $L(v) < L(w)$.\footnote{Such a topological ordering is also referred to as a \emph{weak topological ordering}~\cite{BenderFiGi15} as it does not require a total ordering on the vertices.}  A directed graph has a \defn{cycle} if there exist vertices $u$ and $w$ that are mutually reachable from each other:  that is, $u$ is both an ancestor and descendant of $w$. A topological ordering of a directed graph exists if and only if it is acyclic.

\paragraph{Incremental Graph Problems.}
In the incremental topological ordering and cycle detection problems, initially, there are $n$ vertices $V$ and no edges. The $m$ edges from the set $E$ arrive one at a time and are inserted into the graph data structure.    
Let $G_t$ denote the graph after $t$ edges have been inserted (which we also refer to as \defn{time $t$}).  We assume that after an edge is inserted, the graph continues to be acyclic.  If an edge insertion leads to a cycle, the algorithm must report the cycle and terminate. Thus $G_m$ denotes the final graph (after the last edge insertion that does not create a cycle).  

The performance of the graph data structure is measured as its \defn{total running time} to perform all $m$ edge insertions. In Sections~\ref{sec:dfs} and~\ref{sec:theory}, we use the terms \defn{total cost} and \defn{total update time} and total running time interchangeably.  We the notation $\tilde{O}$ defined as $\tilde{O}(f(n)) = O(f(n) \cdot \text{polylog}(n))$.

\paragraph{Prediction Model.}  In the incremental topological ordering problem with predictions, the data structure additionally obtains a prediction for each vertex $v \in V$ at the beginning. Intuitively, this prediction helps the data structure initialize the label of $v$ to be closer to a feasible topological ordering. 

For a vertex $v$, let $\alpha(v)$ be the total number of ancestor edges of $v$ in the final graph $G_m$.  Analogously, let $\delta(v)$ be the total number of descendant edges of $v$ in the final graph $G_m$.  The Learned-DFS Ordering in Section~\ref{sec:dfs} receives a prediction $\anp(v)$ of $\alpha(v)$ for each vertex $v$.\footnote{Note that the algorithm also works if we instead receive a prediction of the number of ancestor \emph{vertices} of $v$.  However, this increases the running time to $O(m\etamax \frac{m}{n})$.}  The prediction error of a vertex $v$ is $\eta_v = |\anp(v) - \alpha(v)|$.  

The Ideal Learned Ordering in Section~\ref{sec:theory} receives a prediction $\anp(v)$ of the number of ancestors $\alpha(v)$ and $\dep(v)$ of the number of descendants $\delta(v)$ respectively, for each vertex $v$.  The prediction error of the vertex $v$ is ${\eta_v = |\anp(v) - \alpha(v)| + |\dep(v) - \delta(v)|}$.

The overall error is $\etamax = \max_v \eta_v$ throughout the paper.

\section{Learned-DFS Ordering}%
\label{sec:dfs}

In this section, we give a simple and easy-to-implement data structure that achieves $O(m\eta)$ total update time.  We refer to this algorithm as the \defn{Learned DFS Ordering (LDFS)}.

\subsection{Algorithm Description}
At all times, the algorithm maintains a \emph{level} $\ell(v)$ for each vertex, which is a number from $0$ to $m$. For each vertex $v$, the algorithm maintains a linked list $in(v)$ of $v$'s parents at the same level (i.e.\ a linked list of all parents $p$ of $v$ with $\ell(p) = \ell(v)$).  Finally, to maintain a topological ordering, the algorithm additionally maintains a (global) counter $a$, and for each vertex $v$ an integer $j(v)\in\{1, \ldots, nm+n\}$.

Initially, $a = nm+1$, and for each $v$, $\ell(v) = \anp(v)$, $in(v) = \{\}$, and $j(v) = nm+1$.
On insertion of an edge $e = (u, v)$, if $\ell(u) >\ell(v)$, set $\ell(v) \leftarrow \ell(u)$ and $in(v) \leftarrow \{u\}$.  Then, do a forward search from $v$ to recursively update $v$'s descendants.  That is, for each child $w$ of $v$, if $\ell(v) > \ell(w)$, update $\ell(w)$ and $in(w)$ and recurse.  Report a cycle if one is found; otherwise calculate a topological order $T_f$ on all vertices whose levels changed during this search.

\paragraph{Cycle detection.} After the above update concludes, if $\ell(u) = \ell(v)$, 
do a reverse DFS starting at $u$ (i.e.\ a DFS where edges are followed backward) only following edges $in(u)$ from vertices at the same level. 
If this search visits $v$, report a cycle. Otherwise, let $T_b$ be a topological order on the vertices visited during this DFS (e.g., $T_b$ can be computed by ordering vertices in the order of their DFS finish times).

\paragraph{Topological ordering.} The ordering imposed by the level $\ell(v)$ on the vertices is a \defn{pseudo-topological ordering} ~\cite{BenderFiGi15}.  Indeed, our algorithm can be viewed as a simplification of the sparse algorithm in~\cite{BenderFiGi15} with the addition that levels are initialized using predictions.  

Bender et al.\ describe how to extend their ordering to a topological order by breaking ties between vertices on a level using the order in which they are traversed in the reverse DFS.  We use a similar technique here.  Concatenate $T_b$ and $T_f$ to create a single topological order $T$. If $\ell(u) = \ell(v)$ and $j(u) \geq j(v)$, proceed through each vertex $w\in T$ in reverse order.  Set $j(w) = a$, then $a = a-1$, and then set $w$ to the previous vertex in $T$.  

We define the \defn{label} of a vertex $v$ to be $L(v) = \ell(v)(nm+2) + j(v)$.
 The algorithm maintains the following invariants. 

\begin{invariant}
\label{inv:tabels_are_top_sort}
    For any edge $e = (u, v)$ in the graph $G_t$ at time $t$, $\ell(u) \leq \ell(v)$.
\end{invariant}

\begin{invariant}[{\citep[Theorem 2.5]{BenderFiGi15}}]%
\label{inv:a_is_decreasing}
At all times, $a\in\{1, \ldots, nm+1\}$; furthermore, $a$ is nonincreasing over the entire run of the algorithm.
\end{invariant}

\begin{invariant}
\label{inv:label_is_max}
At any time $t$ and any vertex $v$, let $A_t(v)$ be the set of ancestors of $v$ in $G_t$.  Then, the level of $v$ in $G_t$ is
$\ell(v) = \max_{a\in A_t(v)} \anp(a)$.
\end{invariant}

\subsection{Analysis}
The following proves that the algorithm is always correct.

\begin{lemma}%
\label{lem:simple_weak_top_sort}
If the insertion of the last edge creates a cycle in $G_t$, the simple learned algorithm correctly detects and reports it. 
    Furthermore, for any edge $e = (u, v)$ in the graph $G_t$ at time $t$, $L(u) < L(v)$.
\end{lemma}

We bound the running time by bounding the cost of the forward search to update levels, and the reverse DFS within a level to detect a cycle.  

We first upper bound how big the levels can get using $\etamax$.

\begin{lemma}\label{lem:levelupper}
Let $\ell_0$ and $\ell_m$ denote the initial and final level of any vertex $v$.  Then, $\ell_m - \ell_0 \leq 2 \etamax$.
\end{lemma}

\begin{proof}
By Invariant~\ref{inv:label_is_max}, $v$ has some ancestor $u\in G_m$ with $\ell_m(v) = \anp(u)$.  Since $\ell_0(v) = \anp(v)$ by definition, we have that $\ell_m(v) - \ell_0(v) = \anp(u) - \anp(v)$.  Any ancestor of $u$ is also an ancestor of $v$, so $\alpha(v) - \alpha(u) \geq 0$. Thus,
\begin{align*}
\ell_m(v) - \ell_0(v) &= \anp(u) -  \anp(v)  \\
&\leq \anp(u) -  \anp(v)  + (\alpha(v) - \alpha(u)) \\
&= \left(\anp(u) - \alpha(u)\right) +  \left(\alpha(v) - \anp(v)\right) \\
&\leq \etamax + \etamax.\qedhere
\end{align*}
\end{proof}

Lemma~\ref{lem:levelupper} is sufficient to bound the cost of all level updates during the forward search.

\begin{lemma}
\label{lem:simple_cost_update_labels}
    The total cost to update the levels of all vertices is $O(m\etamax)$.
\end{lemma}

\begin{proof}
To obtain the total cost of updating the levels, note that each time we update the level of a vertex $v$, the algorithm recursively updates its children, and then checks each of its parents to update $in(v)$.  This takes $O({\Delta}(v))$ time, where ${\Delta}(v)$ is the sum of the outdegree and indegree of $v$. Thus, using Lemma~\ref{lem:levelupper} the total cost of all level updates is
\begin{align*}
O\left(\sum_v {\Delta}(v) \cdot (\ell_m(v) - \ell_0(v) + 1)\right) = O(m\etamax)~~~~~~~~
\qedhere
\end{align*}
\end{proof}

To bound the cost of the reverse DFS on a level, we bound the number of incoming edges on any level at any time.

\begin{lemma}
\label{lem:simple_ancestors_error}
At any time, if a vertex $v$ has $k$ ancestor edges on its level then $\eta \geq k/2$.
\end{lemma}

Now we can bound the cost of the reverse DFS.
The algorithm maintains incoming edges $in(v)$ of $v$ from vertices on its level in a linked list.
 Performing the reverse DFS from $v$ thus has cost $O(1 + a_t(v))$, where $a_t(v)$ is the number of ancestor edges of $v$ from vertices at level $\ell(v)$ at time $t$. 
By Lemma~\ref{lem:simple_ancestors_error}, $a_t(v_1)\leq \etamax$ and thus the reverse DFS costs $O(\etamax)$ for each insertion.  Finally, combining with Lemma~\ref{lem:simple_cost_update_labels} and the $O(n)$ time for initialization, we get the following result.
\begin{theorem}
The Learned DFS Ordering solves the incremental topological ordering and cycle detection problem with predictions in total running time $O(m\etamax + n)$.
\end{theorem}

 \section{Ideal Learned Ordering}%
\label{sec:theory}

In this section, we give an ideal learned data structure for the incremental topological ordering and cycle detection problem with total update time $\tilde{O}(m + \min\{n\etamax, n^2, m\etamax^{1/3}\})$ for $m$ edge insertions.  We refer to this algorithm as \defn{Ideal Learned Ordering}.

The algorithm receives a prediction  $\tilde{\alpha}(v)$ and $\tilde{\delta}(v)$ of the number of ancestor and descendant edges of each vertex in the final graph $G_m$.  By definition, $|\tilde{\alpha}(v) - \alpha(v)| \leq \etamax$ and $|\tilde{\delta}(v) - \delta(v)| \leq \etamax$ for all $v$.

\paragraph{Prediction Decomposition.}  At a high level, the algorithm decomposes the problem instance into smaller subproblems based on each vertex's prediction, and uses the state-of-the-art worst-case algorithm on each subproblem based on the instance's sparsity.  Recall that BK and BFGT refer to the best-known sparse algorithm by~\cite{BhattacharyaKulkarni20} and the best-known dense algorithm by~\cite{BenderFiGi15}; see Section~\ref{sec:related}.  Using a tighter analysis for these algorithms under predictions, we then bound the running time of each subproblem using the prediction error.

\subsection{Algorithm Description} 
The algorithm maintains an estimate $\heta$ which is an estimate of the overall error $\eta$ based on edges seen so far.  It also maintains a level $\ell(v)$ for each vertex, initialized using both $\anp(v)$ and $\dep(v)$.  It maintains a pseudo-topological ordering over these levels greedily.  
We decompose the initial set of vertices $V$ into a sequence of \defn{subproblems} based on the predictions for each vertex.  When an edge $e = (u,v)$ arrives, it is treated as an edge insertion into each subproblem that contains both $u$ and $v$. 
The algorithm invokes the BK or BFGT algorithm to perform this insertion and to assign internal labels within each subproblem.  

If an edge is inserted across subproblems that violates the ordering over the levels, the algorithm updates its estimate of $\heta$ and rebuilds with an improved decomposition.

 \paragraph{Algorithm setup.}
Let $\heta_i$ be the value of $\heta$ after $i$ edges are inserted.
We begin with $\heta_0 = 1$.  

We maintain a level $\ell(v)$ for each vertex $v$.  Each $\ell(v)$ consists of a pair of numbers: $\ell(v) = (\ell^a(v), \ell^d(v))$; we call this the \defn{ancestor level} and \defn{descendant level} of $v$ respectively.  The idea is that $\ell^a(v)$ and $\ell^d(v)$ are initialized using the predicted ancestors and descendants of $v$ respectively and updated as edges are inserted.

At all times, the level $\ell(v)$ has four possible values satisfying the constraints below.  These are referred to as the \defn{possible levels for $v$}.
\begin{align*}
\ell(v)^a &\in \{
{\lceil \anp(v)/\heta\rceil},
{\lceil \anp(v)/\heta\rceil + 1}
\} \\ 
\ell(v)^d &\in \{
{\lfloor \dep(v)/\heta\rfloor},
{\lfloor \dep(v)/\heta\rfloor - 1}
\} 
\end{align*}
We maintain that for any edge $e = (u, v)$, $\ell^a(u) \leq \ell^a(v)$ and $\ell^d(u) \geq \ell^d(v)$.

At any time, the vertex set $V$ is decomposed into \defn{subproblems}
$H_{j, k}$, where the indices $j, k\in \{0, \ldots, \lceil m/\heta_i \rceil + 1\}$. Each subproblem $H := H_{j,k}$ is a subgraph of $G_t$ and represents an instance of the incremental topological ordering problem (possibly at an intermediate state with some edges already inserted). 
 A vertex $v$ can be part of at most four subproblems, indexed by one of its possible levels:
    \[\calH(v) = \{ H_{j,k} ~|~ (j,k) \text{ is a possible level of $v$}\}.\]

As each vertex is in at most four subproblems, the algorithm maintains $O(n)$ subproblems at any point; note that ``empty'' subproblems are not maintained.

\paragraph{Initialization and Build.} We first describe how to perform a \textsc{Build} on a graph $G_t$; \textsc{Build} is called each time the estimate $\heta$ changes.  At initialization, \textsc{Build}($G_0$) adds each vertex $v$ to the subproblems $\calH(v)$. 
If a sequence of edge insertions $e_1, \ldots, e_t$ are such that $t$th insertion causes $\heta_{t}$ to be updated, then 
\textsc{Build}($G_t$) first updates $\calH(v)$ for each $v$ based on the updated value of $\heta_t$ and adds $v$ to $\calH(v)$. Then it calls $\textsc{Insert}(e_i)$ for $i \in \{1, \ldots, t\}$ using the insert algorithm described next.

The insert algorithm uses a further subroutine \textsc{Build-BFGT}$(H)$, which is used to ``switch'' a subproblem from the sparse case to the dense case.  
Let $V_H$ and $E_H$ be the vertices and edges currently in $H$.
\textsc{Build-BFGT} initializes an instance of BFGT on vertices $V_H$, and then inserts all edges in $E_H$ one by one using BFGT.

\paragraph{Edge Insertion.}  
On the insertion of the $t$th edge $e_t$, \textsc{Insert}$(e_t)$  first recursively updates the ancestor and descendant levels of $v$ and $u$ in $G_t$. That is, if $\ell^a(u) > \ell^a(v)$, set $\ell^a(v) = \ell^a(u)$ and recurse on all out-edges of $v$. Similarly, if $\ell^d(u) < \ell^d(v)$, set $\ell^d(u) = \ell^d(v)$ and recurse on all in-edges of $u$.  This maintains the following invariant.

\begin{invariant}%
\label{inv:levels}
    For any edge $e = (u,v)$, $\ell^a(u)\geq \ell^a(v)$ and $\ell^d(u)\leq \ell^d(v)$.
\end{invariant} 

If for any vertex $v$, the updated value of $\ell(v)$ is not one of the possible levels of $v$, the algorithm doubles $\heta$ (i.e.\ set $\heta_i = 2\heta_{i-1}$) and calls \textsc{Build} on $G_t$.

Next, we describe how the algorithm inserts $e_t$ into all subproblems $H\in \calH(u)\cap\calH(v)$ using the BK or BFGT algorithm based on whether the subproblem is sparse or dense.  Without predictions, a graph is termed sparse if $m = o(n^{3/2})$ and dense otherwise.  To determine if a subproblem with predictions is sparse or dense, the algorithm takes error $\heta$ into account.  
More formally, let $n'$ and $m'$ denote the number of vertices and edges in a subproblem $H$ prior to the insertion of $e_t$ into $H$.  Then:
\begin{itemize}[topsep=1pt, itemsep=1pt]
\item ({\bf Sparse}) If $m' + 1 < n' \eta^{2/3}$, it inserts $e_t$ to the subproblem $H$ using BK;
\item ({\bf Dense}) if $m'>  n' \eta^{2/3}$, it inserts $e_t$ to subproblem $H$ using BFGT;
\item ({\bf Sparse to dense transition}) if $m' < n' \eta^{2/3}$ and $m' + 1 > n' \eta^{2/3}$, it calls \textsc{Build-BFGT}($H$) first, then inserts $e_t$ into $H$ using BFGT.
\end{itemize} 
We refer to the label within a subproblem assigned by the BFGT or BK algorithm as an \defn{internal label} of the vertex.  

 If after $t$ edges are inserted (for any $t$) we have $\heta > n$ and $t\heta^{1/3} > n^2$, the algorithm ignores all predictions and reverts to using the worst-case BFGT.
 The algorithm creates a new instance of BFGT using the vertices in $G_t$, and inserts all $t$ edges into this BFGT instance one by one.   
All future edges are inserted into this BFGT instance. 

\paragraph{Defining Labels.} For any vertex $v$, let $i(v)$ be the internal label of $v$ in subproblem $H_{\ell(v)}$. 
Let $k$ be a positive integer larger than the internal label of any node in a graph with $n$ vertices in either BFGT or BK (we note that both algorithms maintain only nonnegative labels).
Define the label $L$ of $v$ as $L(v) = k(\ell^a(v) + m-\ell^d(v)) + i(v)$.

\subsection{Analysis}
We analyze the correctness and running time of Ideal Learned Ordering.

\paragraph{Correctness.} First, we show that if a cycle exists, then it is correctly reported by the algorithm.  
 
By Invariant~\ref{inv:levels}, if the insertion of an edge creates a cycle, all vertices in the cycle must have the same level $\ell$. The algorithm maintains the invariant that at all times $H_{\ell(v)}\in \calH(v)$, so all vertices and edges in the cycle must be in some subproblem $H$ and thus will be detected by BFGT or BK.

\begin{lemma}\label{lem:correctideal}
    For any edge $e=(u,v)$ in $G_t$,
    $L(u) \leq
    L(v)$.
\end{lemma}
\begin{proof}
    If $\ell^a(u) + m -\ell^d(u) < \ell^a(v)+ m-\ell^d(v)$
    then the lemma holds since $i(u) < k$.

    Otherwise, suppose
    $\ell^a(u) + m -\ell^d(u) \geq \ell^a(v)+ m-\ell^d(v)$.
    By Invariant~\ref{inv:levels}, $\ell^a(u)\leq \ell^a(v)$ and $m-\ell^d(u)\leq m-\ell^d(v)$; thus we must have $\ell(u) = \ell(v)$.  Thus, $i(u)$ and $i(v)$ are both assigned by BFGT or BK on $H_{\ell(u)}$. By correctness of BFGT and BK, $i(u) < i(v)$.
\end{proof}

\paragraph{Running Time Analysis.}
We give an overview of the running time analysis of Ideal Learned Ordering.  Proofs are deferred to Appendix~\ref{sec:omitted_proofs}.

Lemma~\ref{lem:sparse_ancestors} bounds the number of ancestors and descendants of a vertex within the graph of any subproblem.

\begin{lemma}%
    \label{lem:sparse_ancestors}
     For any subproblem $H_{j,k}$ and vertex $v \in H_{j,k}$, $v$ has at most $2 (\heta + \etamax)$ ancestor edges and $2 (\heta +  \etamax)$ descendant edges in $H_{j,k}$.
\end{lemma}

Lemma~\ref{lem:sparse_heta} shows that the estimate $\heta$ maintained by the algorithm is never more than $2 \etamax$.
\begin{lemma}
\label{lem:sparse_heta}
At all times, $\heta \leq 2\etamax$
\end{lemma}

Lemma~\ref{lem:bender_ancestors_running_time} and Lemma~\ref{lem:sparse_ancestors_running_time} bound the cost of running BFGT and BK any subproblem respectively.

\begin{lemma}%
\label{lem:bender_ancestors_running_time}
    Consider a subproblem $H$ with $n'$ vertices and $m'$ edges that are inserted into $H$ one by one.  If each vertex in $H$ has at most $O(\etamax)$ edge ancestors, then the total running time of running BFGT on $H$ is $\tilde{O}(n'\etamax + m')$ time.
\end{lemma}

\begin{lemma}%
\label{lem:sparse_ancestors_running_time}
    Consider a subproblem $H$ with $n'$ nodes and $m'$ edges, such that:
    (1) $m' < \heta^{2}n'/\log^2 n'$, 
        (2)  each vertex in $H$ has at most $O(\etamax)$ edge ancestors and $O(\etamax)$ edge descendants, and
        (3) $\heta = O(\etamax)$.
    Then running BK on $H$ with parameter $\tau =  n^{1/3}\heta^{2/3}/m^{1/3}$ takes total time $\tilde{O}(m' \etamax^{1/3})$ in expectation.
\end{lemma}

Finally, Theorem~\ref{thm:dense_running_time} analyzes the total running time.

\begin{theorem}%
\label{thm:dense_running_time}
Ideal Learned Ordering has total expected running time $\tilde{O}(\min\{m\etamax^{1/3}, n\etamax, n^2\})$.
\end{theorem}

\section{Experiments}\label{sec:experiments}

This section presents experimental results for the Learned DFS Ordering (LDFS) described in Section~\ref{sec:dfs}.  Our experiments show that using prediction significantly speeds up performance over baseline solutions on real temporal data. Moreover, only a small amount of training dataset (e.g., 5\%) is sufficient to see one or two orders of magnitude of improvement.  Finally, we show that LDFS is extremely robust to errors in the predictions.  

Our implementation and datasets can be found at \url{https://github.com/AidinNiaparast/Learned-Topological-Order}.

\paragraph{Algorithms.}
We compare LDFS against two natural baseline solutions that we call DFS I and DFS II. Each of the three algorithms use a greedy depth-first-search approach to maintain a topological ordering, with the difference that LDFS warm-starts its levels using predictions.   

\paragraph{DFS I.} The first algorithm is equivalent to LDFS with zero predictions: that is, $\anp(v) = 0$ for each $v$.

\paragraph{DFS II.} This algorithm was presented by~\cite{marchetti1993line} for incremental topological ordering and revisited by~\cite{franciosa1997incremental} for incremental DFS.  It has total update time $O(mn)$. \cite{baswana2018incremental} perform an empirical study on incremental DFS algorithms and show that DFS II (which they call FDFS) is the state-of-the-art on DAGs. DFS II maintains exactly one vertex at each level. When an edge $(u,v)$ is inserted, if $l(v) < l(u)$, the algorithm performs a partial DFS to detect all the vertices $w$ reachable from $v$ such that $l(v) < l(w) < l(u)$, and updates their levels to be larger than $l(u)$.  

\paragraph{Datasets.} We use real directed temporal networks from the SNAP Large Network Dataset Collection~\cite{snapnets}. To obtain the final DAG $G$, we randomly permute the vertices and only keep the edges that go from smaller to larger positions (this ensures $G$ is acyclic). Then, we sort the edges in increasing order of their timestamps to obtain the sequence of edge insertions.  Table~\ref{table:dataset features} summarizes these datasets.  Note that these graphs are sparse.

\paragraph{Predictions.} To generate the predictions for LDFS, we use a contiguous portion of the input sequence as the training set.  Consider the graph that results from inserting the training set edges into an empty graph.
For each node $v$, we define $\anp(v)$ to be the number of $v$'s ancestor edges in that graph. 

\paragraph{Experimental Setup and Results.}
On real datasets, we compare LDFS to DFS I and II in terms of the number of edges and vertices processed (\defn{cost}) in Table~\ref{table:cost} and in terms of runtime in Table~\ref{table:time}.  The last 50\% of the data in increasing order of the timestamps is used as the test data in all of the experiments in Table~\ref{table:exp}. The training data for LDFS is a contiguous subsequence of the data that comes right before the test data.

We include plots for the email-Eu-core\footnote{https://snap.stanford.edu/data/email-Eu-core-temporal.html}~\cite{paranjape2017motifs} dataset, which contains the email communications in a large European research institution.  A directed edge $(u, v, t)$ in this dataset means that $u$ sent an e-mail to $v$ at time $t$. 
Figure~\ref{fig:email-Eu-core scale training set} shows how the training data size affects the runtime of LDFS. Figure~\ref{fig:email-Eu-core robustness} is
a robustness experiment showing performance versus the noise added to predictions.

For testing robustness to prediction error, we add noise to the predictions. We first generate predictions as described.
Then, we calculate the standard deviation of the prediction error, which we denote by \textsc{SD}(predictions).  Finally, 
we add a normal noise with mean 0 and standard deviation \textsc{SD}(noise) = $C\cdot$ \textsc{SD}(predictions) (for some constant $C$) independently to all of the predictions to obtain our noisy predictions. We repeat the experiment 10 times, each time regenerating the noisy predictions; we plot the mean
and standard deviation of the resulting running time in Figure~\ref{fig:email-Eu-core robustness}.  

\begin{table}[ht]
    \small
    \centering
    \caption{The number of nodes and edges in the real datasets from SNAP. The input sequence has duplicate edges (referred to as temporal edges). The length of the sequence is the number of temporal edges. Static edges are the number of distinct edges.}
    \vspace*{0.1in}
     \begin{tabular}{c|c c c} 
         \textbf{} & Nodes & Static Edges & Temporal Edges 
         \\ [0.5ex] 
         \hline
         \rule{0pt}{2.5ex}
         Email-Eu-core & 918 & 12320 & 171617\\
         CollegeMsg & 1652 & 9790 & 27931\\
         Math Overflow & 14839 & 45267 & 53499\\
     \end{tabular}
    \label{table:dataset features}
\end{table}

\begin{table}[ht]
    \small
    \centering
    \caption{Performance of LDFS against DFS I and II. The test data is the last 50\% of the data. 
    Columns LDFS(5) and LDFS(50) correspond to the performance of LDFS when 5\% and 50\% of the data are used for training, respectively.}
    \vspace*{0.1in}
    \begin{subtable}{\linewidth}
        \small
        \centering
        \begin{tabular}{c|cccc}
             \textbf{} & LDFS(5) & LDFS(50) & DFS I & DFS II 
             \\ [0.5ex] 
             \hline
             \rule{0pt}{2.5ex}
             Email-Eu-core & 8.4e+3 &  2.6e+3 & 5.0e+5 & 3.1e+5\\
             CollegeMsg & 9.6e+3 & 5.4e+3 & 1.2e+5 & 6.8e+5\\
             Math Overflow  & 4.4e+4 &  2.8e+4 & 2.8e+5 & 3.9e+7 \\
        \end{tabular}
        \caption{Cost (\# nodes and edges processed)}
        \label{table:cost}
    \end{subtable}
    \begin{subtable}{\linewidth}
        \small
         \centering
         \begin{tabular}{c|cccc}
             \textbf{} & LDFS(5) & LDFS(50) & DFS I & DFS II 
             \\ [0.5ex] 
             \hline
             \rule{0pt}{2.5ex}
             Email-Eu-core & 0.071 & 0.078 & 0.274 & 0.226\\
             CollegeMsg & 0.021 & 0.016 & 0.101 & 0.336\\
             Math Overflow & 0.094 & 0.078 & 0.241 & 18.373\\
         \end{tabular}
         \caption{Running Time (s)}
         \label{table:time}
    \end{subtable}
\label{table:exp}
\end{table}

\begin{figure}[ht]
    \centering
    \begin{subfigure}{\linewidth}
        \centering  \includegraphics[width=.9\linewidth, height=3.6cm]{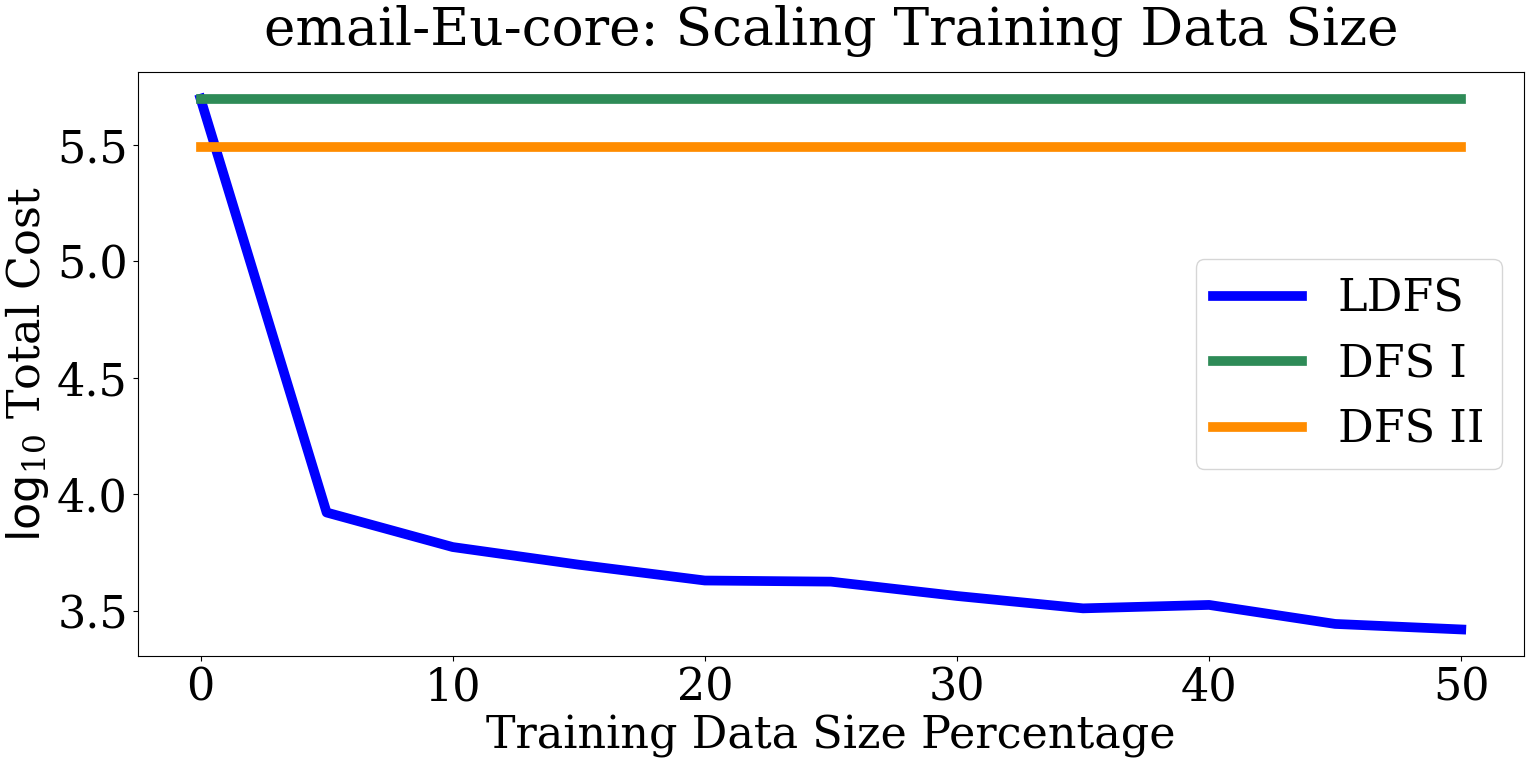}
        \caption{}
        \label{fig:email-Eu-core scale training set}
    \end{subfigure}
    \begin{subfigure}{\linewidth}
        \centering \includegraphics[width=.9\linewidth, height=3.6cm]{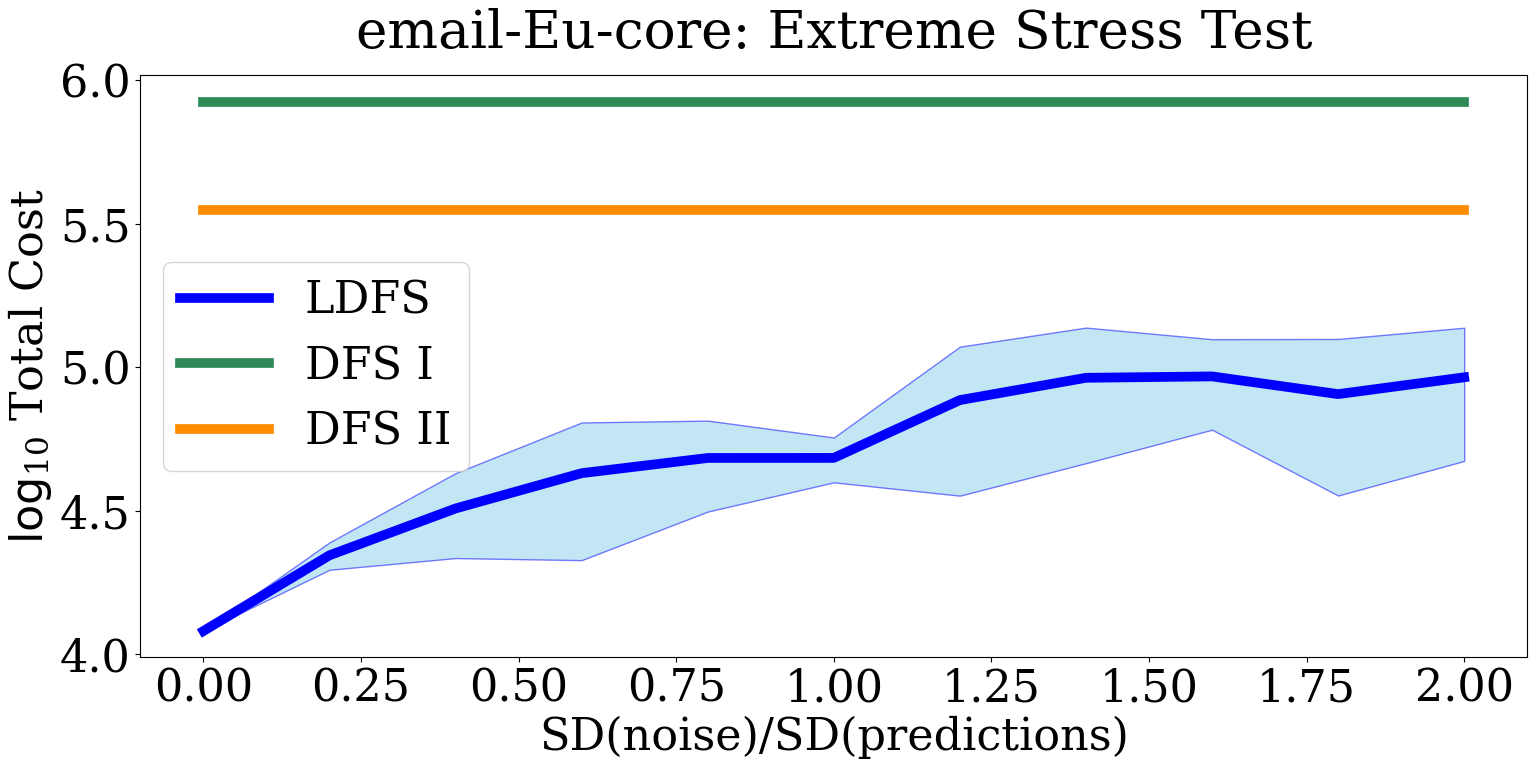}
        \caption{}
        \label{fig:email-Eu-core robustness}
    \end{subfigure}
    \caption{Total cost (number of nodes and edges processed) of LDFS compared to the two baselines for email-Eu-core dataset, in logarithmic scale.
    In Figure~\ref{fig:email-Eu-core scale training set}, the x-axis is the percentage of the input sequence used as training data for LDFS. 
    Figure~\ref{fig:email-Eu-core robustness} shows the effect of adding noise to predictions on the cost of LDFS. The first 5\% of the input is used as the training data and the last 95\% as the test data. For different values of $C$, a normal noise with mean 0 and standard deviation (\text{SD}) of $C\cdot$SD(predictions) is independently added to each prediction. This noise is regenerated 10 times.  The x-axis is \text{SD}(noise)/\textsc{SD}(predictions). The blue line is the mean and the cloud around it is the SD of these experiments.}
    \label{fig:email-Eu-core}
\end{figure}

\paragraph{Discussion.} Results in Table~\ref{table:exp} demonstrate that, in all cases, even a very basic prediction algorithm can significantly enhance performance over the baselines. 
Only 5\% of historical data is needed to see a significant difference between our methods's performance and the baselines; in some cases up to a factor of 36 in cost. Better predictions obtained from 50\% of historical data improve performance further, up to a factor of 116.

Finally, Figure~\ref{fig:email-Eu-core robustness} shows that LDFS is very robust to bad predictions. 
For example, note that if $\text{SD(noise)} \geq 2\cdot\text{SD(predictions)}$, then $\approx 61\%$ of the noisy predictions have noise added to them that is at least SD(predictions)\footnote{In a normal distribution, $\approx 61\%$ of items are more than half a standard deviation from the mean.}---thus, the relative value of the predictions becomes largely random for many items.  Since the LDFS algorithm's performance only
depends on how predictions for different nodes relate to each other (not their value), this represents a significant amount of noise, effectively nullifying the predictions of many nodes.
Nonetheless, LDFS still outperforms the baselines even for this extreme stress test.  Moreover, increasing the noise degrades the performance confirming that the efficiency does depend on the quality of predictions. 

In Appendix~\ref{sec:additional_experiments}, we include additional plots for the datasets in Table~\ref{table:dataset features}.  We also investigate the effect of edge density on
performance for synthetic DAGs.  These experiments further support our conclusions; in particular, even for very dense DAGs, our algorithm still outperforms the baselines, although with smaller margins.

\section{Conclusion}\label{sec:conclusion}

This paper gave the first dynamic graph data structure that leverages predictions to maintain an incremental topological ordering.   We show that the data structure is ideal: that is, it is consistent, robust, and smooth with respect to errors in prediction. Thus,  predictions deliver speedups on typical instances while never performing worse than the state-of-the-art worst case solutions.  This paper is also the first empirical evaluation of using predictions on dynamic graph data structures. 
 Our experiments show that the theory is predictive of practice: predictions deliver up to $6$x speedup for LDFS compared to natural baselines.

Our results demonstrate the incredible potential for improving the theoretical and empirical efficiency of data structures using predictions. It would be interesting to explore how predictions can be leveraged for designing data structures for other dynamic graph problems.

% PUT THIS BACK FOR CAMERA READY IF ACCEPTED
%\section*{Impact Statement}
%This paper presents work whose goal is to advance the field of Machine Learning and Algorithms. There are many potential societal consequences of our work, none which we feel must be specifically highlighted here.

\balance
\bibliography{topsort}
\bibliographystyle{icml2024}

\newpage
\appendix

\section{Omitted Proofs}
\label{sec:omitted_proofs}

\begin{proofof}{Lemma~\ref{lem:simple_weak_top_sort}}
By Invariant~\ref{inv:tabels_are_top_sort}, for any cycle $C$ in the graph, all vertices in $C$ must at the same level.  Each time we add an edge $e = (u, v)$, if $\ell(u) = \ell(v)$, the algorithm checks whether the addition of this edge creates a cycle within that level through a reverse depth-first search. 

    Now, assume there is no cycle in $G_t$; we show that a weak topological sort is maintained.  A weak topological sort is trivially maintained in $G_0$, so assume inductively that the algorithm correctly maintains a weak topological sort in $G_{t-1}$.
    Consider an edge $(u,v)\in G_t$.  If $\ell(u) < \ell(v)$, then the label of $u$ is less than the label of $v$ since $i(u),i(v) \leq nm+1$.  If $\ell(u) = \ell(v)$, then we split into cases based on if the label of $u$ or $v$ was changed during the updates after the $t$th edge was inserted.  If neither $u$ or $v$ were updated, the labels continue to be a topological ordering as in $G_{t-1}$. It is not possible that $v$ is updated but $u$ is not: for any $v$ visited during DFS, since $\ell(u) = \ell(v)$, $u$ is also visited; for any $v$ whose label is updated, $u$ must have a strictly larger label than any other parent of $v$. If $u$ is updated and $v$ is not, then 
        $i(u)$ is set to $a$; since $a$ decreases each time some $i(w)$ is set, we must have $i(u) < i(v)$. If both $u$ and $v$ are updated, 
        $u$ must come before $v$ in $T$.
        Again, since $a$ decreases each time some $i(w)$ is set, we must have $i(u) < i(v)$.
\end{proofof}

\begin{proofof}{Lemma~\ref{lem:simple_ancestors_error}}  
Let $A$ denote the set of all ancestors of $v$ at level $\ell(v)$ at the current time.  
Consider the vertices in $A$ after all edges are inserted (in $G_m$): since $G_m$ is acyclic, there must be at least one vertex $z\in A$ such that no vertex $w \in A$ has that $w$ is an ancestor of $z$ in $G_m$.  

Since $z$ is an ancestor of $v$, all ancestor edges of $z$ are ancestor edges of $v$. However by definition of $z$, an ancestor edge of any $w\in A$ is never an ancestor edge of $z$.  All $k$ ancestor edges of $v$ on its level are ancestor edges of some $w\in A$. Therefore, $\alpha(v) \geq \alpha(z) + k$, so 
$\alpha(v) - \alpha(z) \geq k$.

As levels only increase $\ell(v) \geq \anp(v)$.  By Invariant~\ref{inv:label_is_max}, $\ell(v) \leq \anp(z)$;  
equivalently, $-\anp(v) + \anp(z) \geq 0$.  Summing the above two inequalities, we get 
\[
(\alpha(v) - \anp(v)) +  (\anp(z) - \alpha(z)) \geq k.
\]
Thus, either $\eta_v \geq k/2$ or $\eta_z \geq k/2$.
\end{proofof}

\begin{proofof}{Lemma~\ref{lem:sparse_ancestors}}
Let $H$ refer to the subgraph $H_{j,k}$ after the last edge is inserted into it (thus, $H$ includes edges that are inserted in the future, whereas $H_{j,k}$ does not).    We use $\alpha^H(v)$ and $\delta^H(v)$ to denote the number of 
number of ancestor and descendant edges of a vertex $v$ in $H$.
  
Let $u$ be an ancestor of $v$ in $H$, such that no ancestor edge of $v$ in $H$  is an ancestor edge of $u$ in $H$. Such a $u$ always exists as $H$ is acyclic and can be found by recursively following in-edges of $v$.

By definition of $u$, all ancestor edges of $u$ are ancestor edges of $v$ (in $G_m$); however, no ancestor edges of $v$ in $H$ are ancestor edges of $w$ (in $G_m$).  Thus, $\alpha(v) \geq \alpha(u) + \alpha^H(v)$, so $\alpha(v) - \alpha(u) \geq a^H(v)$.

As both $v$ and $u$ are in $H_{j,k}$ we can bound the difference of their predictions using $\heta$.  That is, 
$j \geq \lceil \anp(v)/\heta\rceil$, and therefore 
$j\geq \anp(v)/\heta$.
Similarly, 
$j\leq \lceil \anp(u)/\heta\rceil + 1$, so 
$j \leq \anp(u)/\heta + 2$.  
Combining, $\anp(u)/\heta + 2 \geq \anp(v)/\heta$, so $\anp(u) - \anp(v) \geq -2\heta$.

Summing the two above equations, we obtain that 
\[
\left(\alpha(v) - \anp(v) \right) + 
\left(\anp(u) - \alpha(u) \right) \geq a^H(v) - 2\heta
\]

By the definition, $\alpha(v) - \anp(v) \leq \etamax$ and $\anp(u) - \alpha(u) \leq \etamax$.  Substituting, $a^H(v) \leq 2\heta + 2\etamax$.

The analysis for the number of descendants is analogous.
 Let $w$ be a descendant of $v$ in $H$, such that no descendant edge of $v$ in $H$  is a descendant edge of $w$ in $H$ . 
By definition of $w$, all descendant edges of $w$ are descendant edges of $v$ (in $G_m$); however, no descendant edges of $v$ in $H$ are descendant edges of $w$ (in $G_m$).  Therefore, $\delta(v) \geq \delta(w) + \delta^H(v)$, so $\delta_m(v) - \delta(w) \geq \delta^H(v)$.

As both $v$ and $w$ are in $H_{j,k}$, we have that 
$j \leq \lfloor \dep(w)/\heta\rfloor$, and therefore 
$j\leq \dep(w)/\heta$.
Similarly, 
$j\geq \lfloor \dep(v)/\heta\rfloor - 1$, so 
$j \geq \dep(v)/\heta - 2$.  
Combining, $\dep(w)/\heta \geq \dep(v)/\heta - 2$, so 
$\dep(w) - \dep(v) \geq -2\heta$.

Summing the two above equations, we obtain that 
\[
\left(\delta(v) - \dep(v) \right) + 
\left(\dep(w) - \delta(w) \right) \geq \delta^H(v) - 2\heta.
\]

By the definition, $\delta(v) - \dep(v) \leq \etamax$ and $\dep(w) - \delta(w) \leq \etamax$.  Substituting, $\delta^H(v) \leq 2\heta + 2\etamax$.
As the number of ancestor and descendant edges are nondecreasing,  this upper bound (in $H$ after all edges are inserted), is also an upper bound at all times in $H_{j,k}$.
\end{proofof}

\begin{proofof}{Lemma~\ref{lem:sparse_heta}}
 We proceed by induction. The lemma is trivially satisfied at time $0$ (since $\heta_0 = 1$), as well as any time where $\heta$ does not change.

   Consider a time when $\heta$ is increased, from $\heta$ to $2\heta$; we show that $\heta \leq 2\etamax$.  When $\heta$ is increased, there is some vertex $v$ with $\ell(v) \notin \calH(v)$.
  We split into two cases based on if the ancestor level or the descendant level constraint is violated: $\ell^a(v) > \lceil \anp(v)/\heta\rceil + 1$, and $\ell^d(v) < \lfloor \dep(v)/\heta\rfloor - 1$.   We begin with the first case.
   Without loss of generality, consider a vertex $v$ that violates the constraint such that no ancestor of $v$ violates the constraint.  Specifically, 
  $\ell^a(v) > \lceil \anp(v)/\heta\rceil + 1$, whereas $\ell(u) \leq \lceil \anp(v)/\heta\rceil + 1$ for all ancestors $u$ of $v$.

   When inserting an edge $e = (x,y)$, the algorithm updates the ancestor levels of all descendants of $x$ to have the same ancestor levels as $x$; no other ancestor levels are updated.  
   Thus, $v$ has an ancestor $w$ with $\anp(w) = \ell^a(v)$.

   Noting that the label of $v$ can only increase, we must have that $\anp(w) = \ell(v) > \lceil \anp(v)/\heta\rceil + 1$.  
   Thus, $\anp(w) - \anp(v) >  \heta$.

   Since $w$ is an ancestor of $v$, $\alpha(w) < \alpha(v)$, so $\alpha(v) - \alpha(w) \geq 0$.  Summing the above two equations,
   \[
   (\anp(w) - \alpha(w)) + (\alpha(v) - \anp(v)) > \heta
   \]
   Thus either $\eta_w > \heta/2$ or $\eta_v > \heta/2$, so $\heta < 2\etamax$.

   The analysis for the descendant constraint is identical.
\end{proofof}

\begin{proofof}{Lemma~\ref{lem:bender_ancestors_running_time}}
   For each vertex, BFGT maintains a vertex level (that determines the internal label for our algorithm), and a vertex count.  
   In the proof of~{\citep[Theorem 3.6]{BenderFiGi15}}, each edge traversal in BFGT increases the vertex level or a vertex count, and the running time of BFGT is upper bounded by the number of edge traversals plus $m'$ (i.e.\ $O(1)$ additional time for each inserted edge, even if no edge is traversed).  Thus, our goal is to bound the number of times a vertex level or vertex count increases in $H$.

   A vertex level begins at 0 and is nondecreasing for all vertices by definition.  By~\citep[Theorem 3.5]{BenderFiGi15}, the level of each vertex is upper bounded by the number of (vertex) ancestors, which is in turn upper bounded by the number of edge ancestors.  
   Since each vertex has $O(\eta)$ vertex ancestors by Lemma~\ref{lem:bender_ancestors_running_time}, the total number of vertex level increases is $\tilde{O}(\etamax)$, giving $\tilde{O}(n'\etamax)$ increases overall.

   Next, we summarize how a vertex count changes over time, and use this to show that it increases by the maximum vertex level.  Let $\ell = \tilde{O}(\etamax)$ be the maximum vertex level of any vertex.
   See the proof of~\citep[Theorem 3.6]{BenderFiGi15} for more details.  
   The data structure maintains a parameter $j$ for each vertex $v$, where $j$ is at most $\log_2 (\text{current vertex count of $v$})$.
   The count for a vertex $v$ begins at $0$, and increases up to $3\cdot 2^j$, after which it is reset to $0$.  This count must increase by at least $2^j$ over the same time.  Thus, so far, the number of times a vertex count is incremented is at most $3\ell$.
    The count may be reset to $0$ one additional time (at most $3\ell$ more increases); furthermore, the count may at the end of the algorithm increase up to $3\cdot 2^j$ without being reset (another $3\ell$ more increases).  Thus, a vertex count can be incremented at most $9\ell$ times.
\end{proofof}

\begin{proofof}{Lemma~\ref{lem:sparse_ancestors_running_time}}
 The cost of BK as shown in~\cite{BhattacharyaKulkarni20} is
\[
\tilde{O}\left(m'n'/\tau + n'^2/\tau + \sqrt{m'^3\tau/n'} + \sqrt{m'n'\tau}\right).
\]

 First, we show that if all vertices in $H$ have at most $O(\etamax)$ edge ancestors and $O(\etamax)$ edge descendants, then the running time of BK on $H$ is 
\begin{equation}
\label{eq:desired_sparse_black_box}
\tilde{O}\left(\frac{m'\etamax}{\tau} + \frac{n'\etamax}{\tau} + \sqrt{\frac{m'^3\tau}{n'}} + \sqrt{m'n'\tau}\right).
\end{equation}

Let us begin with the first term of Equation~\ref{eq:desired_sparse_black_box}.  This term comes from~\citep[Lemma 2.2]{BhattacharyaKulkarni20}.  Specifically, there are $n/\tau$ sampled vertices in expectation; we maintain all ancestors and descendants of each sampled vertex.  This can be done efficiently using the classic data structure presented in~\citep{Italiano86}.  

The result as stated in~\cite{Italiano86} states that the descendants of \emph{all} (rather than just sampled) vertices can be maintained in $O(nm)$ time. Our results require a slightly stronger analysis.\footnote{In fact, \cite{BhattacharyaKulkarni20} also need a stronger analysis, simpler to that presented here, since they only maintain the descendants of sampled vertices.}  For completeness, we summarize this tighter analysis here.  The bounds in~\cite{Italiano86} are based on a potential function analysis, where each vertex $v$ has (using the notation of~\cite{Italiano86}) potential $-(|\text{vis}(x) + 3|\text{desc}(x))$, where vis$(x)$ is the number of descendant edges of $x$, and desc$(x)$ is the number of descendant vertices of $x$.  
They show that their amortized cost (the cost plus the change in potential) of an edge insert is $O(1)$, and that the potential of all nodes is nonincreasing.
We observe that if we only want to maintain the descendants of sampled vertices, we can set the potential of non-sampled nodes to $0$; their amortized analysis argument still holds under this change.  By Lemma~\ref{lem:sparse_ancestors} and Lemma~\ref{lem:sparse_heta}, the potential of any node is at least $-6\etamax$, so the total cost to maintain the descendants of each sampled vertex is $O(\etamax)$.  
An essentially-identical analysis shows that the total cost to maintain all ancestors of sampled nodes is $O(\etamax)$.
Since there are $n/\tau$ expected sampled nodes, we obtain a total expected cost of $O(n\etamax/\tau)$.

Now, the second term of Equation~\ref{eq:desired_sparse_black_box}.  In~\citep[Lemma 2.3]{BhattacharyaKulkarni20}, it is shown that the total time in ``phase II'' is $\tilde{O}(n^2/\tau)$.  In short, they show that the cost for a vertex $v$ is $\tilde{O}(A_S(v) + D_S(v))$, where $A_S(v)$ and $D_S(v)$ are the number of sampled ancestor and descendant vertices of $v$ respectively.  Since a vertex is sampled with probability $\Theta(\log n/\tau)$, they obtain expected cost $\tilde{O}(n/\tau)$ per vertex.  A vertex in $H$ has only $O(\etamax)$ ancestor or descendant edges, and therefore only $O(\etamax)$ ancestor or descendant vertices, and therefore expected cost $\tilde{O}(\log n' \etamax/\tau)$.  Summing over all $n'$ vertices of $H$ we obtain the desired second term.

The third and fourth term of Equation~\ref{eq:desired_sparse_black_box} remain unchanged; thus the running time of BK on $H$ is given by Equation~\ref{eq:desired_sparse_black_box}.

Substituting $\tau = n'^{1/3}\heta^{2/3}/m'^{1/3}$, we obtain running time $m\etamax^{1/3}$.  
Note that BK samples vertices with probability $\Theta(\log n/\tau)$, so we need that $\tau = \Omega(\log n)$.  This is satisfied for large $n'$ due to 
    $m' < \heta^{2}n'\log^2 n'$.  We note that if BK was to sample vertices with a fixed probability $C_1\log n'/\tau$, we could replace the final $\log n'$ term in our bound on $m'$ with $C_1$.
\end{proofof}
\begin{proofof}{Theorem~\ref{thm:dense_running_time}}
We bound the cost of updating the levels first; then we bound the total cost of all subgraphs.

 First, we consider the cost of updating levels after the $i$th edge is inserted.  We only traverse an edge $(u, v)$ while updating levels if the level of $u$ is updated.
    
    First, consider an update when $\heta$ does not increase; thus each vertex has one of its possible levels after the update.
    Each vertex has four possible levels, so each vertex can have its levels updated once per value of $\heta$; thus, each edge can only be traversed once per value of $\heta$.  This leads to $\tilde{O}(m\log \heta_m)$ time.
    
    Now, the other case: if $\heta$ increases, the cost of the scan is at most $O(m)$; since $\heta$ increases $\log_2 \heta$ times, this gives an additional $\tilde{O}(m\log \heta_m)$ time.  

Now, the cost of inserting all edges into their corresponding subgraphs.  Let us begin with some observations about the cost of a single subgraph $H_{i,j}$ with $n'$ vertices and (after all insertions are complete) $m'$ edges, for a fixed $\heta$.  If $m' < \heta^{2/3} n'/\log^2 n$, then by Lemma~\ref{lem:sparse_ancestors_running_time} (note that $m' < \heta^{2/3} n'/\log^2 n$ implies $m' < \heta^2 n'/\log^2 n'$), all edge insertions into $H_{i,j}$ cost $m'\heta^{1/3}$.  If $m' \geq \heta^{2/3} n'/\log^2 n'$, then the first $\heta^{2/3}n'/\log^2 n'$ insertions into $H_{i,j}$ have cost $\tilde{O}(\heta n')$ by Lemma~\ref{lem:sparse_ancestors_running_time}.  All remaining insertions (including reinserting the first $\heta^{2/3}n'/\log^2 n'$ edges during \textsc{Rebuild}) have cost $O(n'\heta)$ by Lemma~\ref{lem:bender_ancestors_running_time}, for $O(n'\heta)$ total time.  Overall, all edge insertions into $H_{i,j}$ take $O(\min\{n'\heta, m'\heta^{1/3}\})$ time.

Now, we sum over all subgraphs and over all values of $\heta$ to achieve the final running time.  
Let $\ell_{\eta} = \log_2 \heta$; thus when $\heta$ doubles $\ell_{\eta}$ is incremented.  Let $n_{i,j,\heta}$ and $m_{i,j,\heta}$ be respectively the number of vertices and total number of edges in $H_{i,j}$ under a given $\heta$.  Then we can bound the total time spent in all subgraphs as:
\begin{multline*}
\sum_{\ell_\eta = 0}^{\lceil \log_2 \eta\rceil + 1} 
\sum_{i=0}^{m/\heta+1}
\sum_{j=0}^{m/\heta+1}
\tilde{O}(\min\{\heta n_{i,j,\heta}, \heta^{1/3} m_{i,j,\heta}\}) \leq\\
\min\left\{
\sum_{\ell_\eta = 0}^{\lceil \log_2 \eta\rceil + 1} 
\sum_{i=0}^{m/\heta+1}
\sum_{j=0}^{m/\heta+1}
\tilde{O}(\heta n_{i,j,\heta}),\right.\\
\left.
\sum_{\ell_\eta = 0}^{\lceil \log_2 \eta\rceil + 1} 
\sum_{i=0}^{m/\heta+1}
\sum_{j=0}^{m/\heta+1}
\tilde{O}(\heta^{1/3} m_{i,j,\heta}).
\right\}
\end{multline*}
Since each vertex is in at most $4$ subgraphs,
\[
\sum_{i=0}^{m/\heta+1}
\sum_{j=0}^{m/\heta+1} n_{i,j,\eta} \leq 4n.
\]
and
\[
\sum_{i=0}^{m/\heta+1}
\sum_{j=0}^{m/\heta+1} m_{i,j,\eta} \leq 4m.
\]
Substituting, the total running time on all subgraphs is 
$\tilde{O}(\min\{n\eta, m \eta^{1/3}\})$.

If at any time $\heta > n$ and $m\heta^{1/3} > n^2$ we stop the above process and use BFGT.  The cost of all edge inserts while $\heta \leq n$ is $\tilde{O}(n^2)$ by the above; the cost of all remaining inserts is $\tilde{O}(n^2)$~\cite{BenderFiGi15}.  Thus, the overall total running time is
$ \tilde{O}(\min\{n\eta, m \eta^{1/3},n^2\})$.
\end{proofof}

\section{Additional Experiments}
\label{sec:additional_experiments}
In this section, we present additional experiments; in particular, we explore how the performance is influenced by the edge density of the graph in synthetic DAGs. 
We also further describe the experimental setup and the datasets we use.

\paragraph{Dataset Description.} 
Here we describe the real temporal datasets we use in our experiments.
\begin{itemize}
    \item email-Eu-core\footnote{https://snap.stanford.edu/data/email-Eu-core-temporal.html}~\cite{paranjape2017motifs}: This network contains the records of the email communications between the members of a large European research institution.  A directed edge $(u, v, t)$ means that person $u$ sent an e-mail to person $v$ at time $t$.
    \item CollegeMsg\footnote{https://snap.stanford.edu/data/CollegeMsg.html}~\cite{panzarasa2009patterns}: This dataset includes records about the private messages sent on an online social network at the University of California, Irvine.  A timestamped arc $(u, v, t)$ means that user $u$ sent a private message to user $v$ at time $t$.
    \item Math Overflow\footnote{https://snap.stanford.edu/data/sx-mathoverflow.html}~\cite{paranjape2017motifs}: This is a temporal network of interactions on the stack exchange web site Math Overflow\footnote{https://mathoverflow.net/}.  We use the answers-to-questions network, which includes arcs of the form $(u,v,t)$, meaning that user $u$ answered user $v$'s question at time $t$.
\end{itemize}

\begin{figure}[h]
    \centering
    \begin{subfigure}{\linewidth}
        \centering \includegraphics[width=\linewidth, height=4cm]{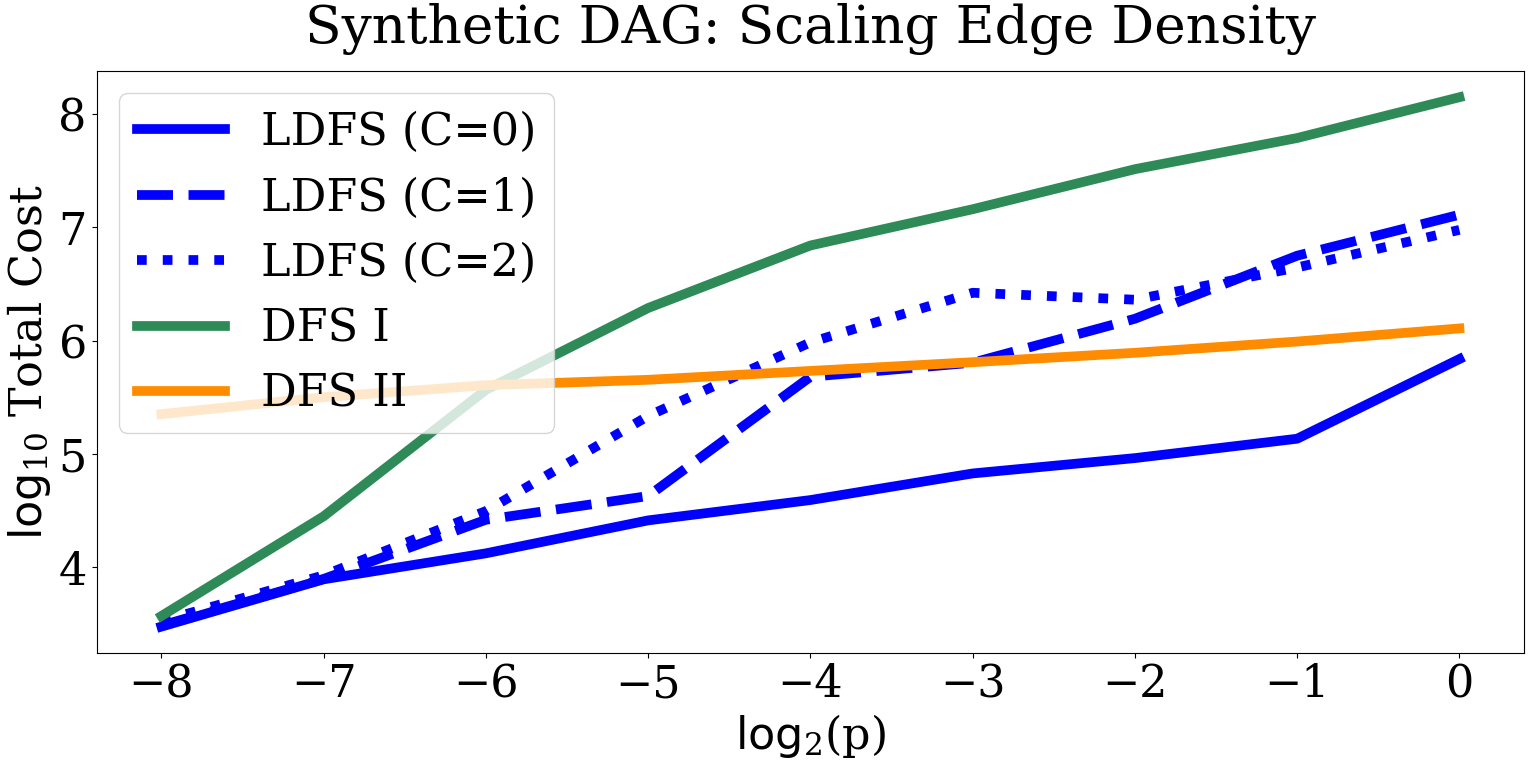}
        \vspace*{-0.2in}
        \caption{}
    \label{fig:DAGcost}
    \end{subfigure}
    \par\smallskip
    \begin{subfigure}{\linewidth}
        \centering \includegraphics[width=\linewidth, height=4cm]{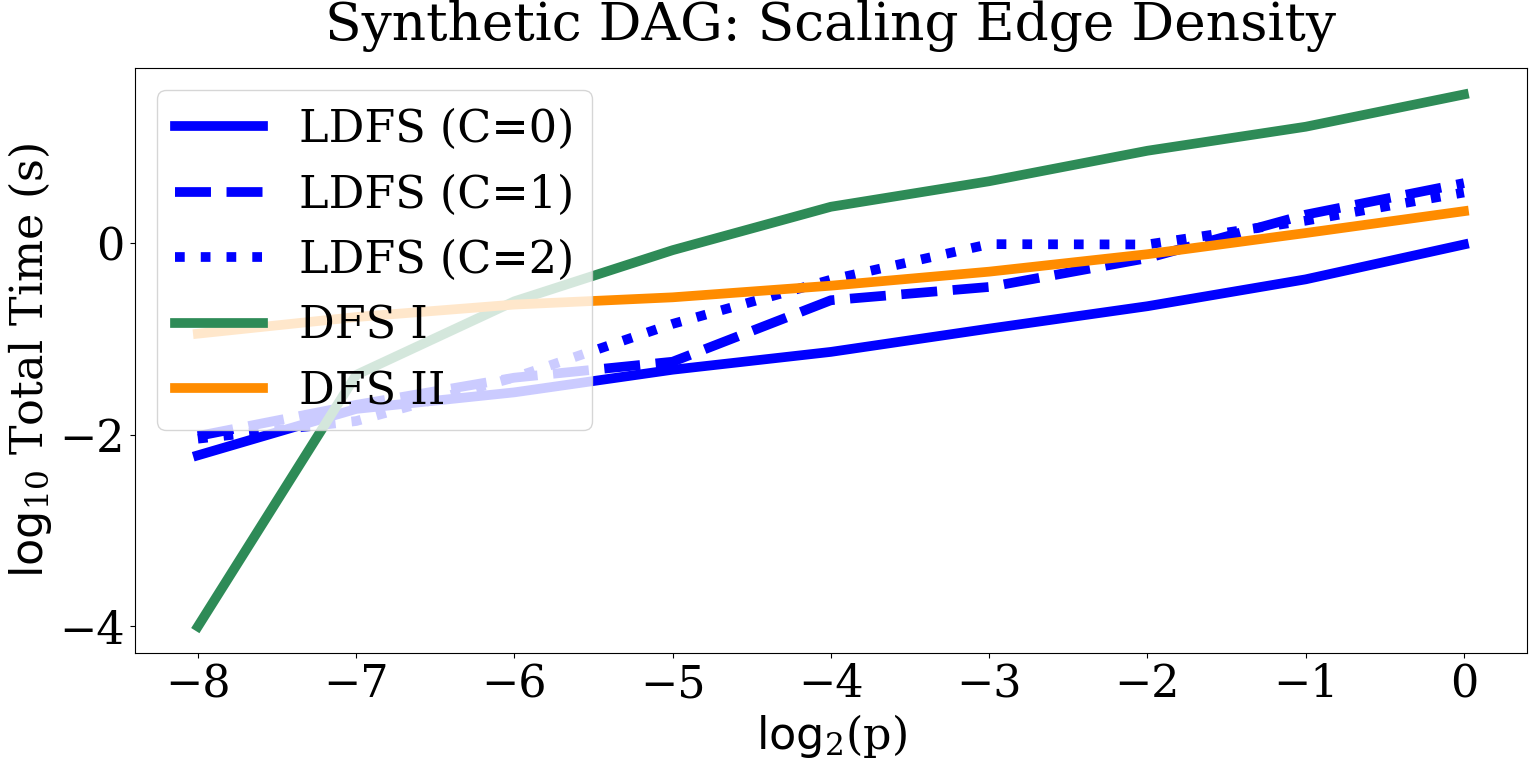}
        \vspace*{-0.2in}
        \caption{}
    \label{fig:DAGtime}
    \end{subfigure}
    \vspace*{-0.3in}
    \caption{Performance comparison for different edge densities on synthetic DAGs (in logarithmic scale). The number of nodes is $n=1000$, and we increase $p$ in the x-axis (in logarithmic scale).  We use the first 5\% of the input as the training data for LDFS (our algorithm), and the last 95\% is used as the test data for all the algorithms. The blue lines correspond to the results for LDFS, with different amounts of perturbation added to the predictions.  The perturbation is a normal noise with mean 0 and standard deviation $C.\text{SD(predictions)}$ that is independently added to each prediction, where SD(predictions) is the standard deviation of the initial predictions.  We include the results for $C=0,1,2$.  The blue lines are the average of 5 different runs, each time regenerating the noise.  Figures~\ref{fig:DAGcost} and~\ref{fig:DAGtime} illustrate the cost (number of nodes and edges processed) and the runtime of these experiments, respectively.}
    \label{fig:DAG}
\end{figure}

\paragraph{Experimental Setup and Results.} We use Python 3.10 on a machine with 11th Gen Intel Core i7 CPU 2.80GHz, 32GB of RAM, 128GB NVMe KIOXIA disk drive, and 64-bit Windows 10 Enterprise to run our experiments. Note that the cost of the algorithms, i.e., the total number of edges and nodes processed, is hardware-independent. 

The datasets we use might include duplicate arcs, but both our algorithm and the baselines skip duplicate edges, both in the training phase and the test phase. 
To check if an arc already exists in the graph, we use the set data structure in Python, which has an average time complexity of $O(1)$ for the operations that we use. 

We use a random permutation of the nodes for the initial levels of the nodes in the DFS II algorithm. For all the experiments on this algorithm, we regenerate this permutation 5 times and report the average of these runs. 

To generate the synthetic DAGs for the experiments on the edge density, we set $V=\{1,\ldots,n\}$, and for each $1 \leq u < v \leq n$, we sample the edge $(u,v)$ independently at random with some (constant) probability $p$.  
We randomly permute the edges to obtain the sequence of inserts. 

Figure~\ref{fig:DAG} compares the performance of LDFS and the two baselines on synthetic DAGs with $n=1000$ nodes and different edge densities. Only the first 5\% of the data is used as the training set for LDFS, and the rest is used as the test set. We also show the effect of adding a huge perturbation to the predictions.  Importantly, we show that the quality of predictions is essential to our algorithm's performance: for very dense graphs, and sufficient additional noise added to the predictions, our algorithm's performance degrades to be worse than the baseline.

\begin{figure}[h]
    \centering
    \begin{subfigure}{\linewidth}
        \centering \includegraphics[width=\linewidth, height=4cm]{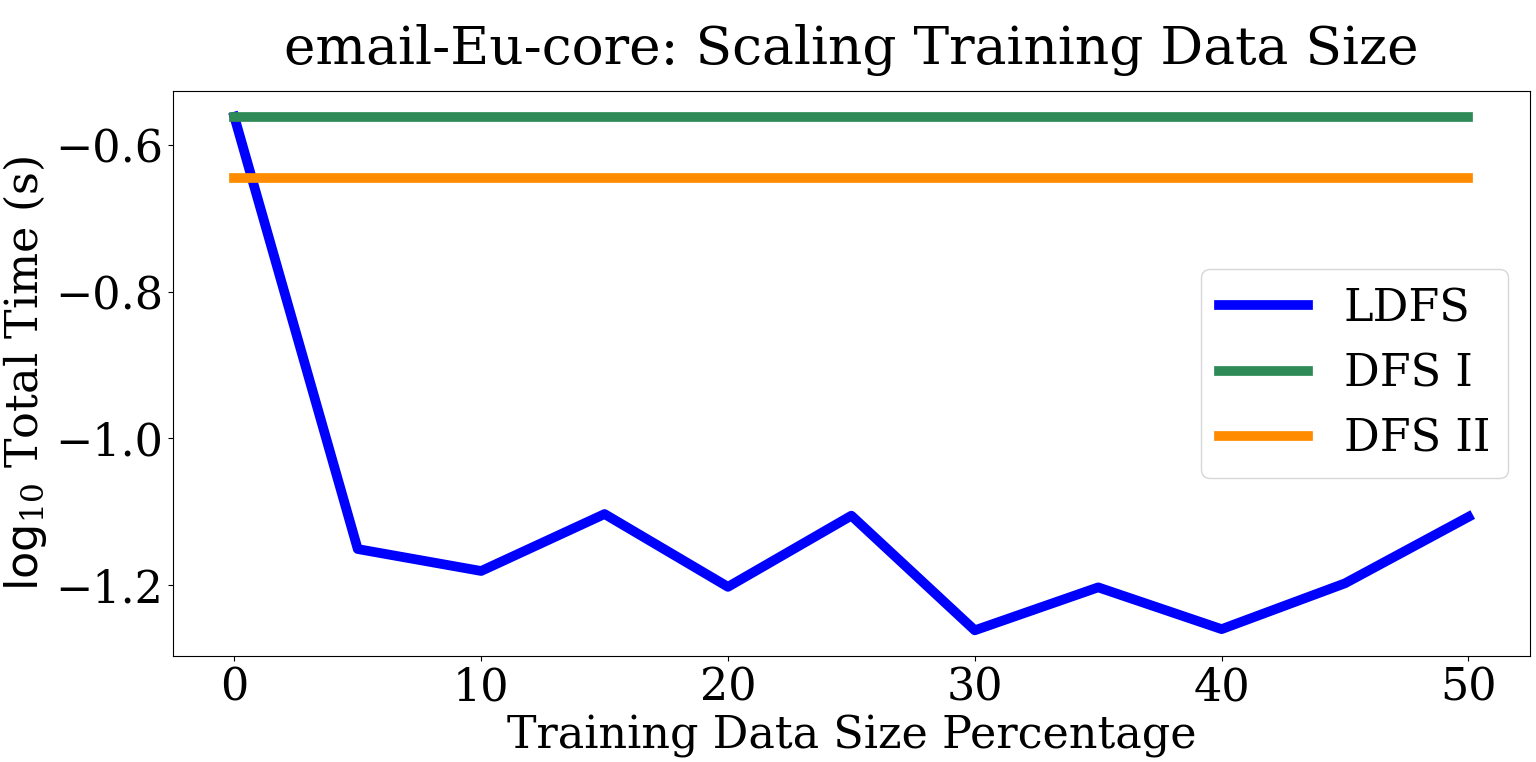}
        \vspace*{-0.25in}
        \caption{}
    \label{fig:email-Eu-core-scale-training-set-time}
    \end{subfigure}
    \par\smallskip
    \begin{subfigure}{\linewidth}
        \centering \includegraphics[width=\linewidth, height=4cm]{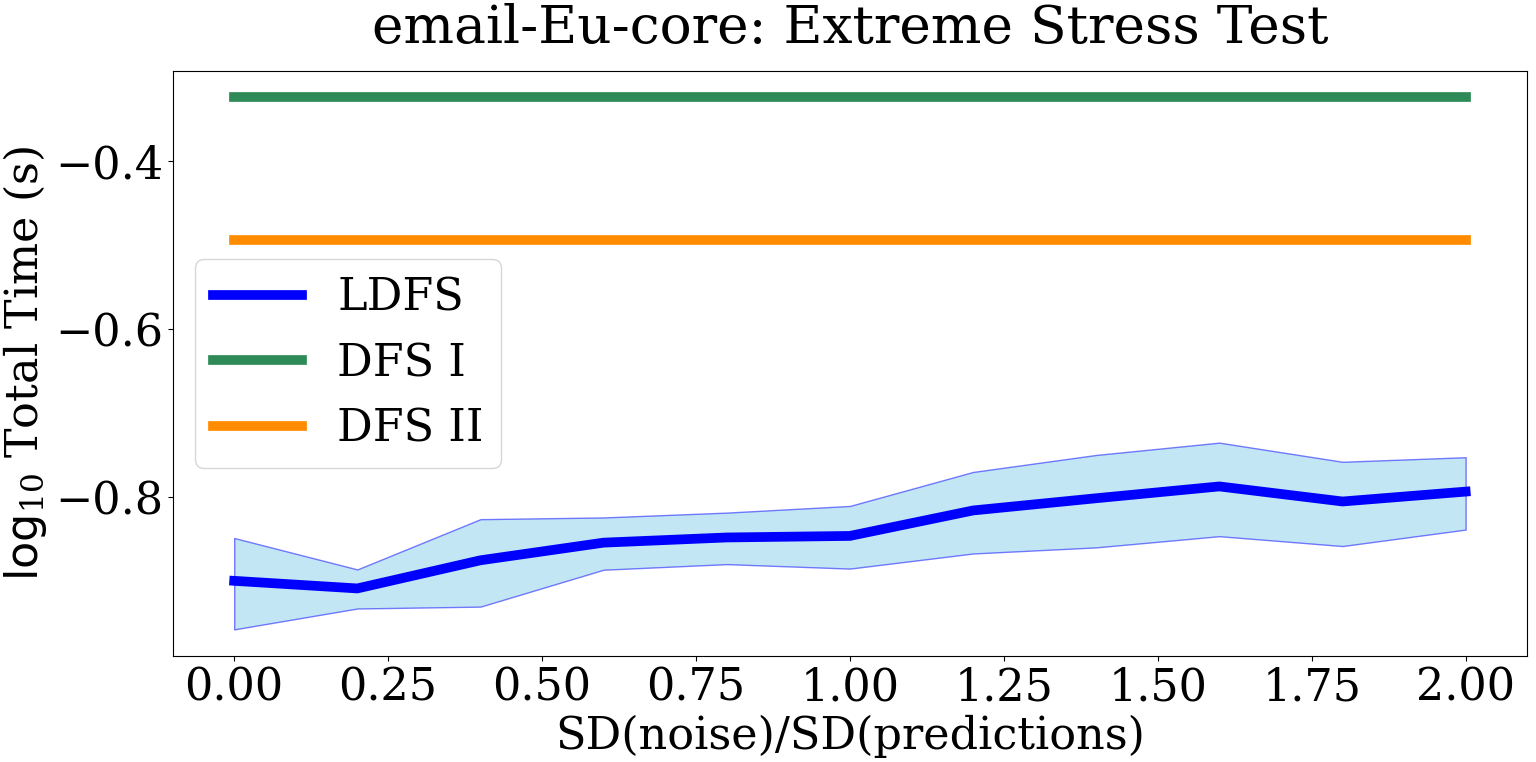}
        \vspace*{-0.25in}
        \caption{}
    \label{fig:email-Eu-core-robustness-time}
    \end{subfigure}%
    \caption{email-Eu-core}%
    \label{fig:email-Eu-core-time}
\end{figure}

In Figure~\ref{fig:email-Eu-core-time}, we show the runtime plots for email-Eu-core. The setup is the same as that of Figure~\ref{fig:email-Eu-core}, except that here we measure the runtime instead of the cost. Figures~\ref{fig:CollegeMsg} and~\ref{fig:MathOverflow} show the same experiments for the other two datasets in Table~\ref{table:dataset features}.

\begin{figure}[h!]
    \centering
    \begin{subfigure}{\linewidth}
        \centering  \includegraphics[width=\linewidth, height=4cm]{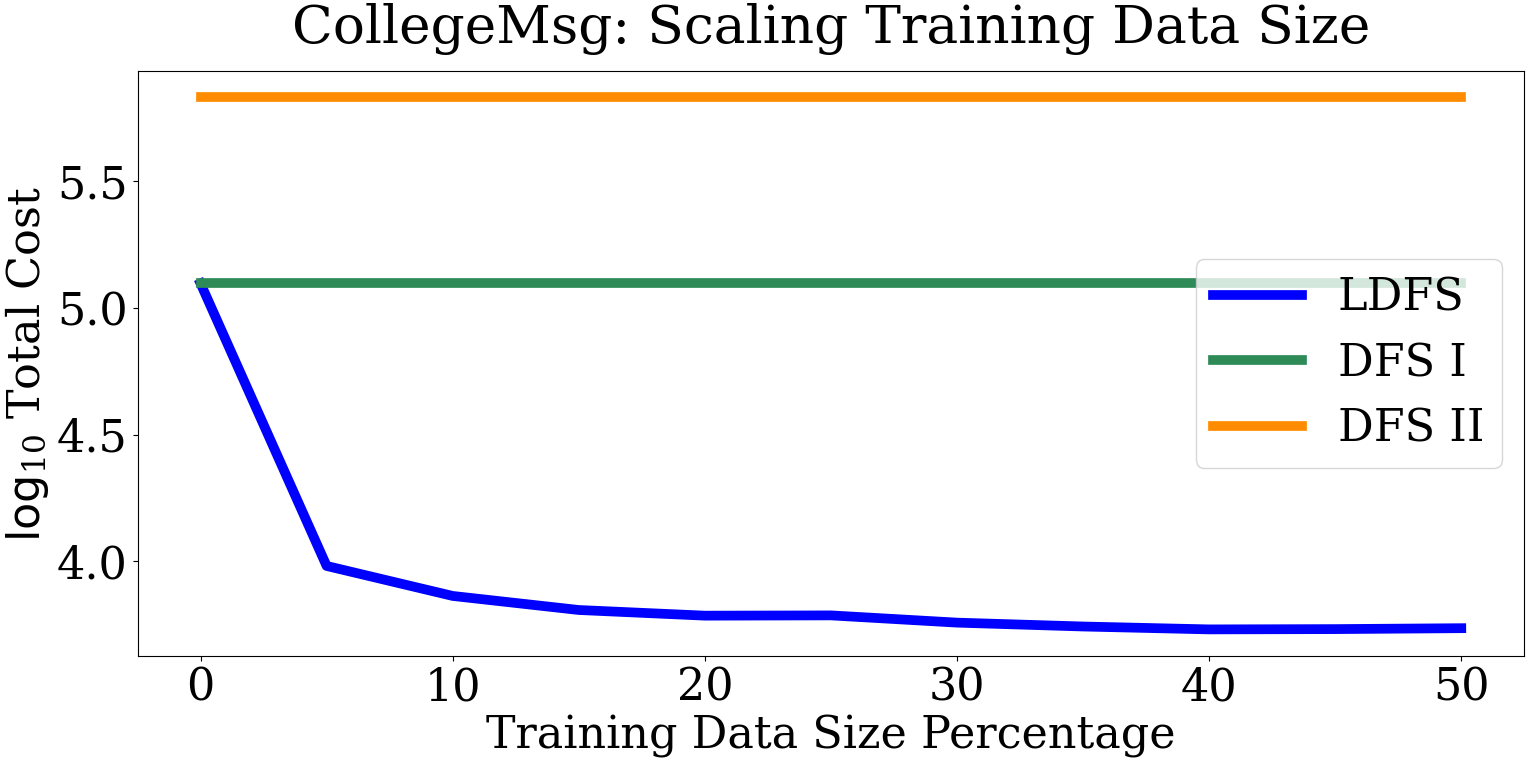}
        \vspace*{-0.2in}
        \caption{}
        \label{fig:CollegeMsg scale training set cost}
    \end{subfigure}
    \par\smallskip
    \begin{subfigure}{\linewidth}
        \centering \includegraphics[width=\linewidth, height=4cm]{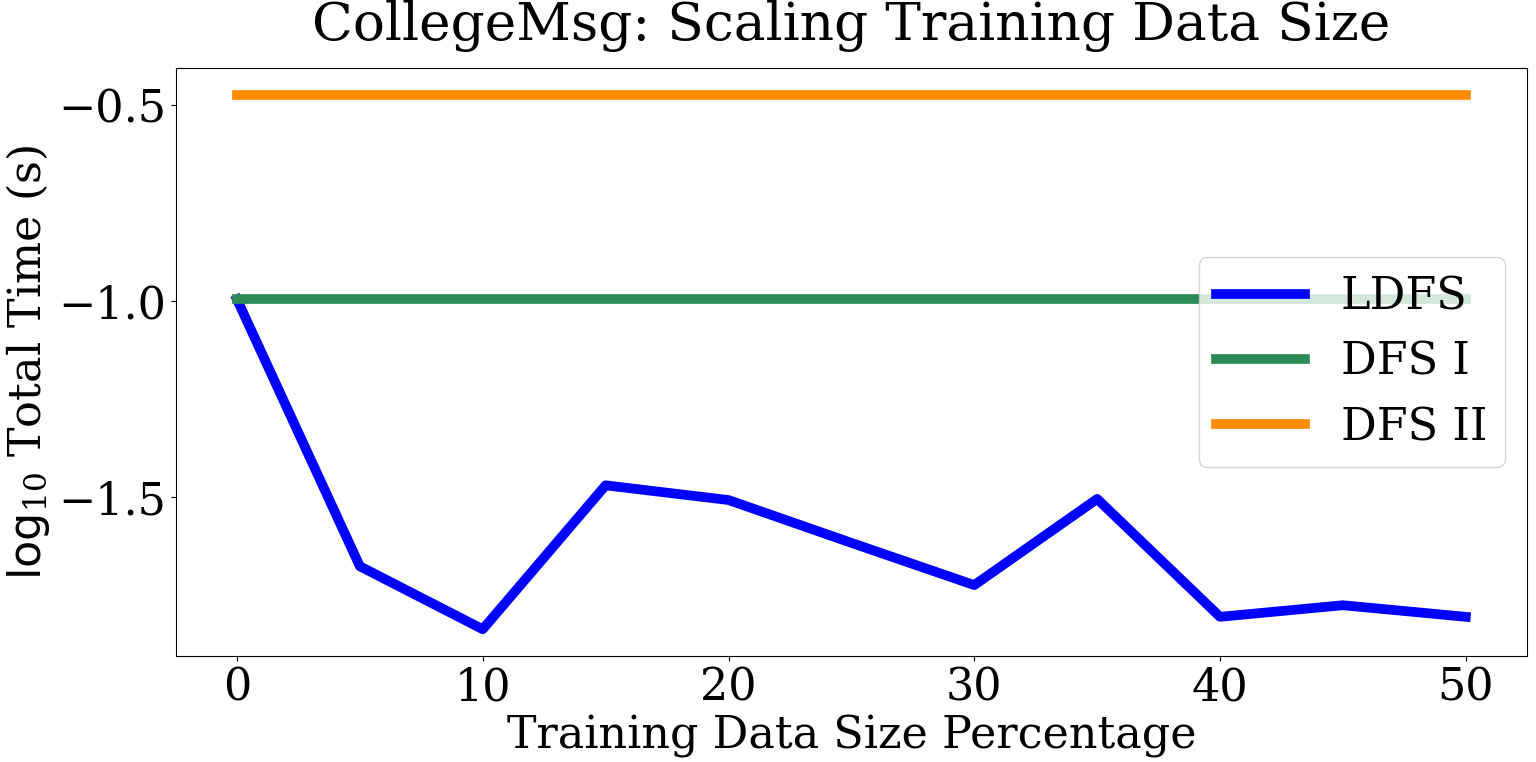}
        \vspace*{-0.2in}
        \caption{}
    \label{fig:CollegeMsg scale training set time}
    \end{subfigure}
    \par\smallskip
    \begin{subfigure}{\linewidth}
        \centering \includegraphics[width=\linewidth, height=4cm]{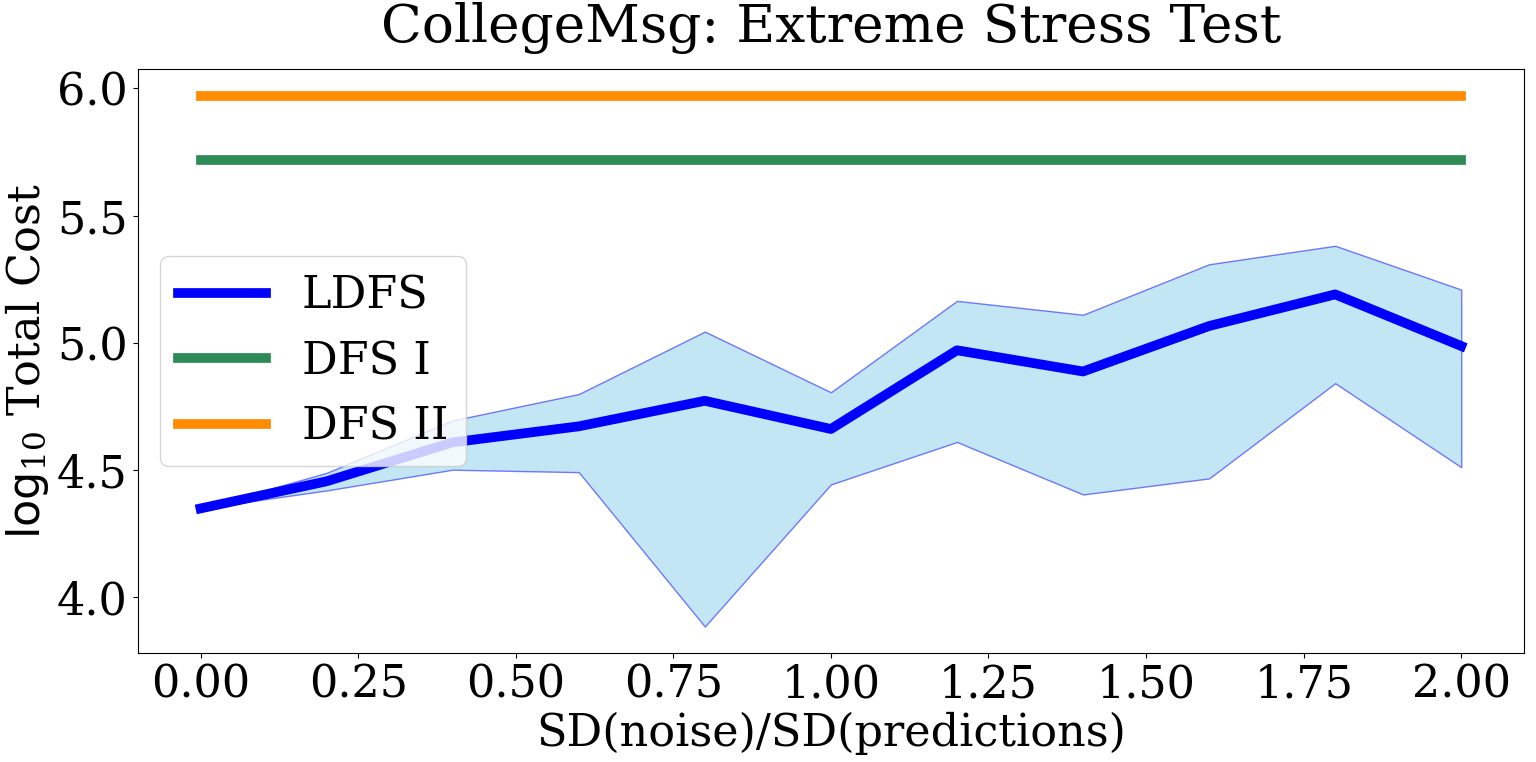}
        \vspace*{-0.2in}
        \caption{}
    \label{fig:CollegeMsg robustness cost}
    \end{subfigure}
    \par\smallskip
    \begin{subfigure}{\linewidth}
        \centering \includegraphics[width=\linewidth, height=4cm]{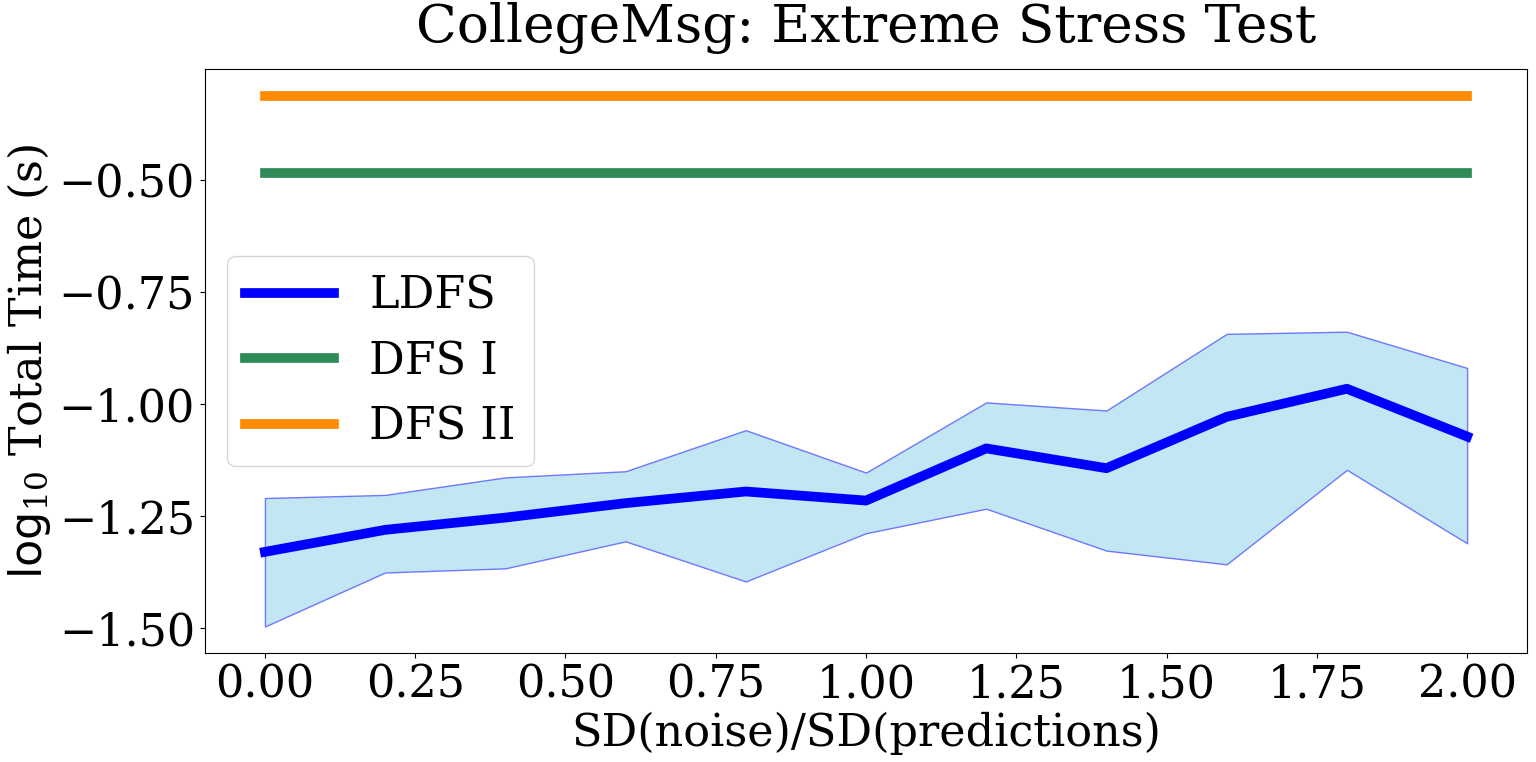}
        \vspace{-0.2in}
        \caption{}
    \label{fig:CollegeMsg robustness time}
    \end{subfigure}
    
    \caption{CollegeMsg}
    \label{fig:CollegeMsg}
\end{figure}

\paragraph{Discussion.} Figure~\ref{fig:DAG} suggests that our algorithm (without perturbation) outperforms the baselines, even for very dense DAGs (note that the last point in the x-axis corresponds to $p=1$, which means that the DAG is complete). However, as the edge density of the DAG increases, the gap between our algorithm and DFS II decreases. Also for high densities and high perturbations, our algorithm still performs reasonably compared to other baselines in terms of cost (which is the main focus of the paper). Another observation is that LDFS is more robust to perturbation on sparse graphs. 
Finally, Figures~\ref{fig:email-Eu-core-time},~\ref{fig:CollegeMsg}, and~\ref{fig:MathOverflow} further support our conclusions in Section~\ref{sec:experiments}.

\begin{figure}[h!]
    \centering
    \begin{subfigure}{\linewidth}
        \centering  \includegraphics[width=\linewidth, height=4cm]{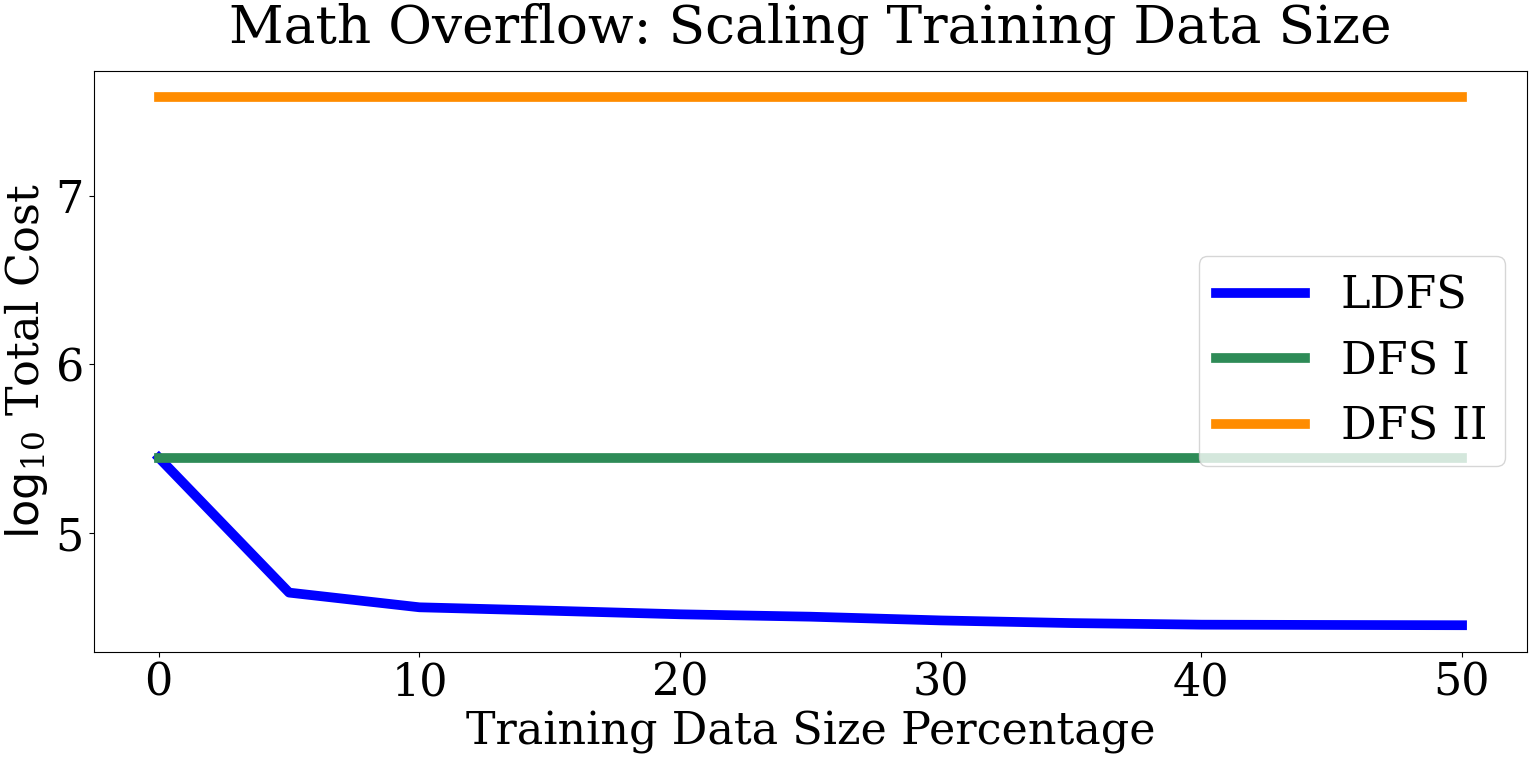}
        \caption{}
        \label{fig:MathOverflow scale training set cost}
    \end{subfigure}
    \par\smallskip
    \begin{subfigure}{\linewidth}
        \centering \includegraphics[width=\linewidth, height=4cm]{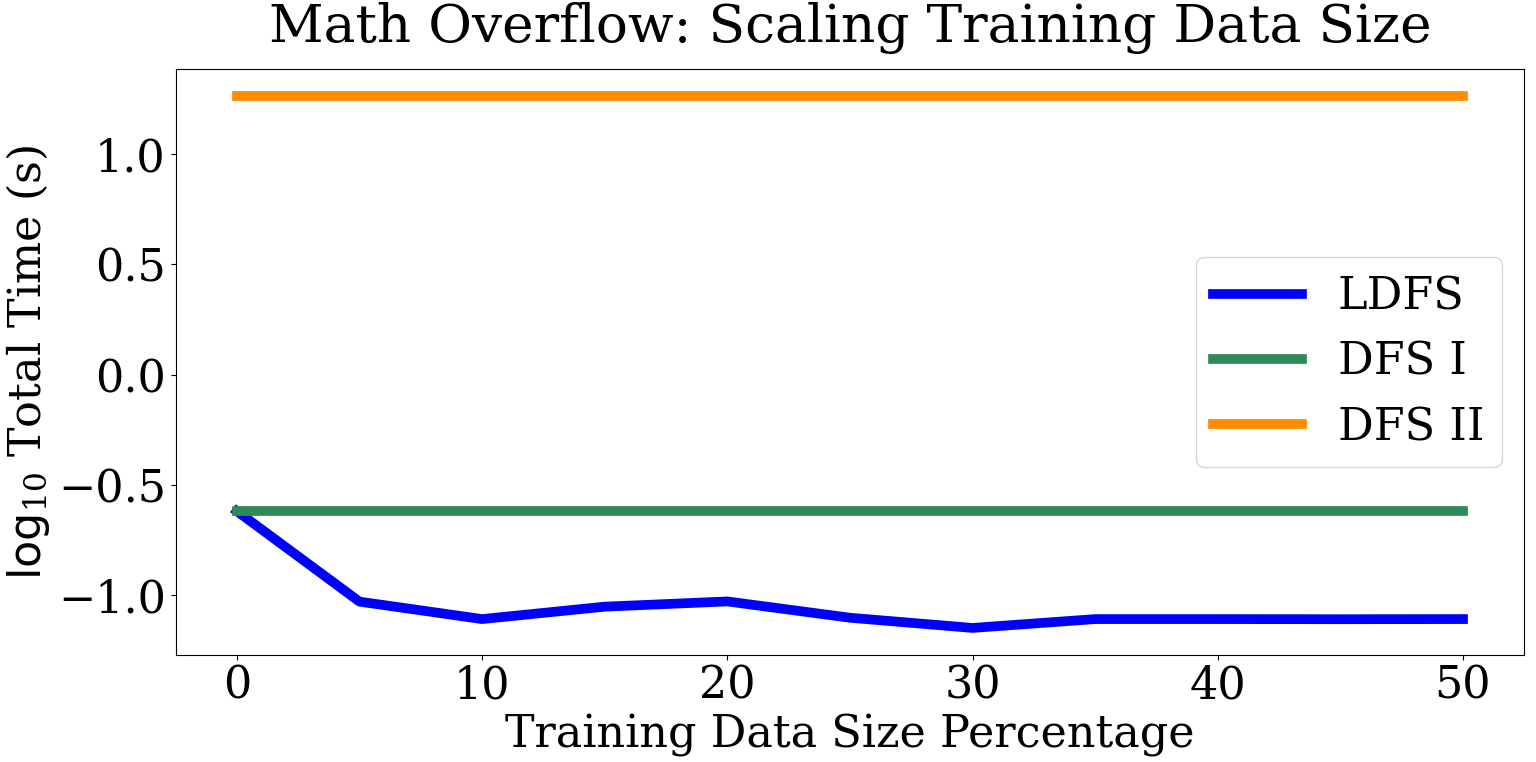}
        \caption{}
    \label{fig:MathOverflow scale training set time}
    \end{subfigure}
    \par\smallskip
    \begin{subfigure}{\linewidth}
        \centering \includegraphics[width=\linewidth, height=4cm]{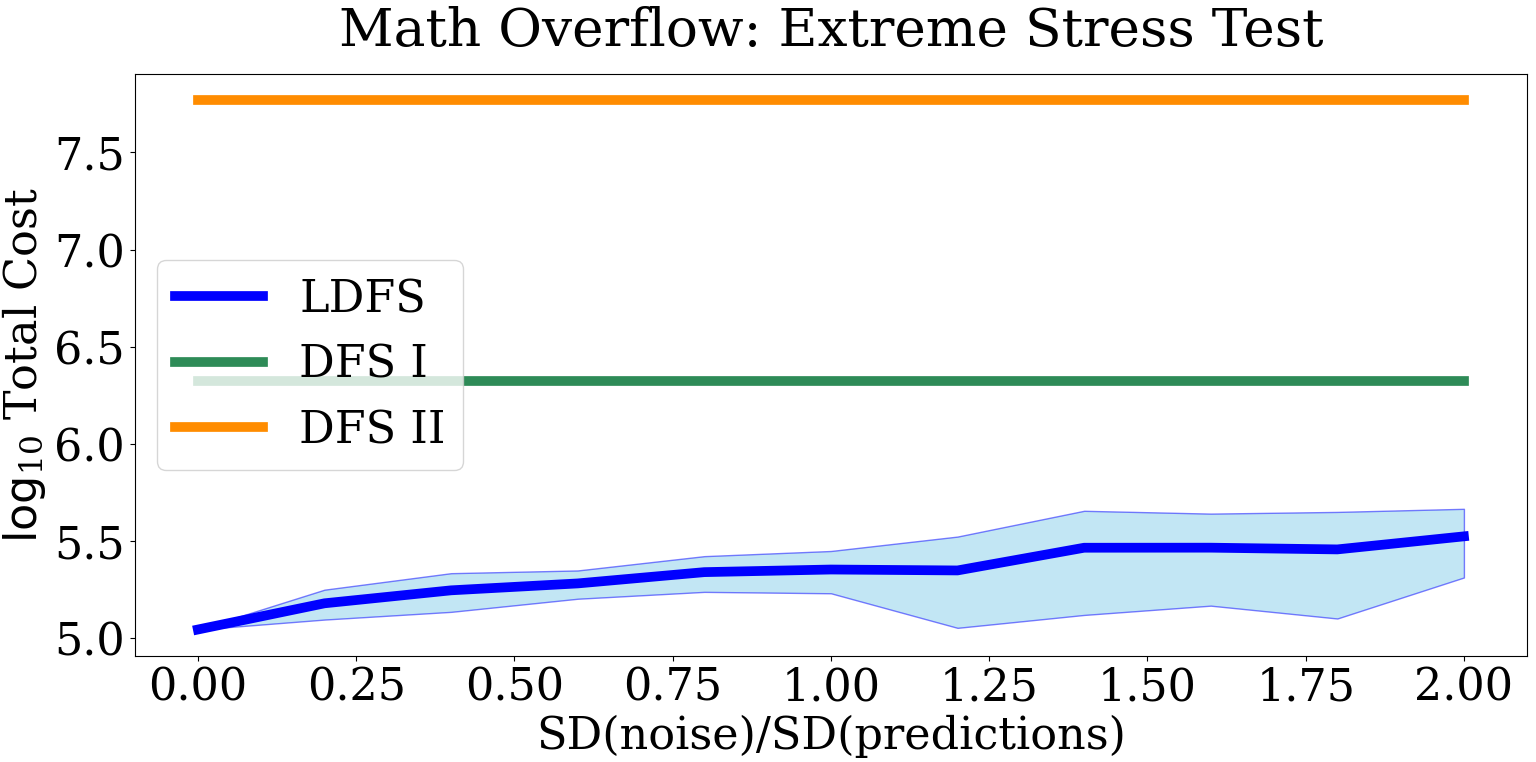}
        \caption{}
    \label{fig:MathOverflow robustness cost}
    \end{subfigure}
    \par\smallskip
    \begin{subfigure}{\linewidth}
        \centering \includegraphics[width=\linewidth, height=4cm]{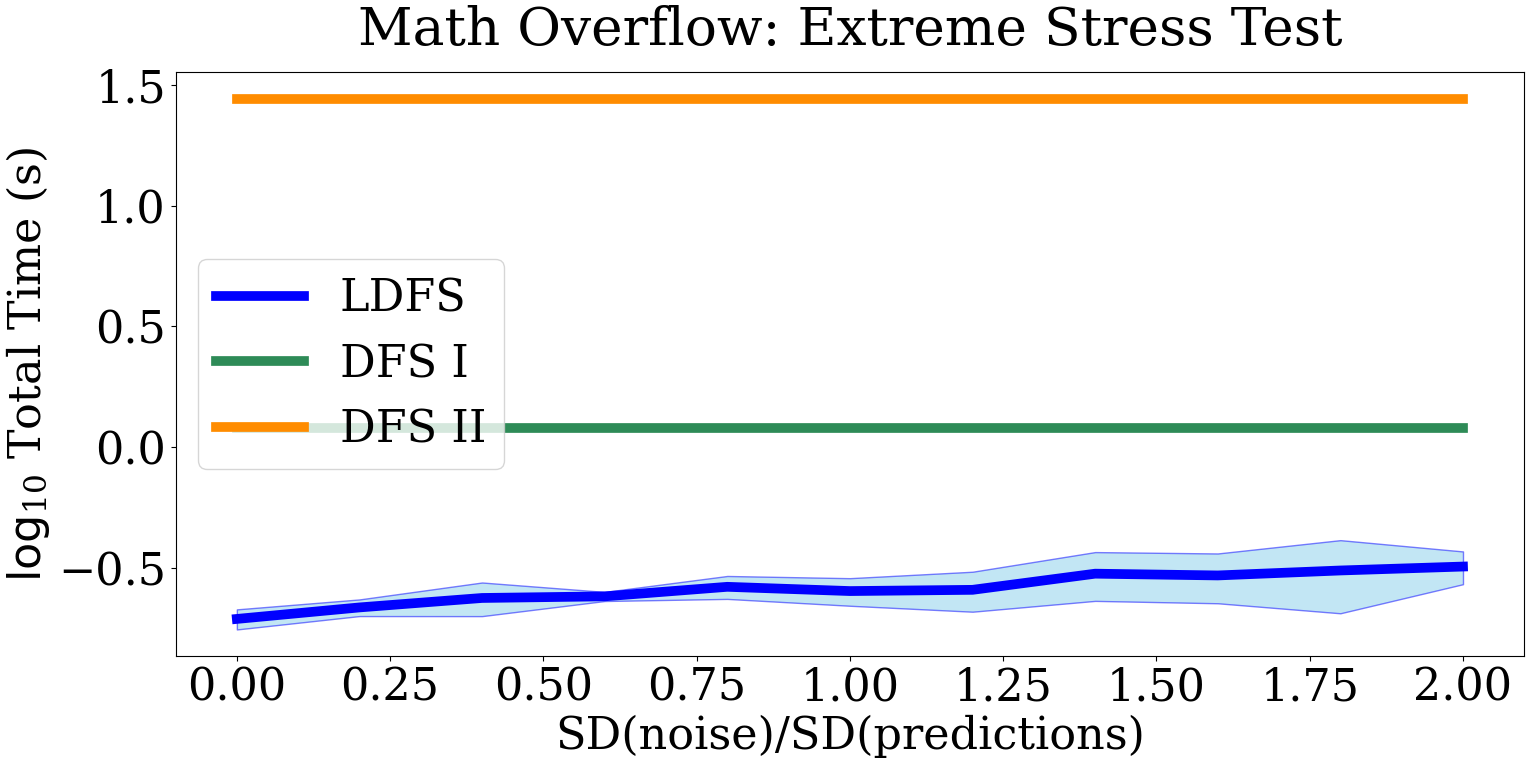}
        \caption{}
    \label{fig:MathOverflow robustness time}
    \end{subfigure}
    \caption{Math Overflow}
    \label{fig:MathOverflow}
\end{figure}

\onecolumn

\end{document}